\documentclass[a4paper,UKenglish,cleveref, autoref, thm-restate]{lipics-v2021}

\hideLIPIcs  

\title{Completely Independent Steiner Trees} 

 \author{Anil Maheshwari}{Carleton University, Ottawa, Canada}{anil@scs.carleton.ca}{https://orcid.org/0000-0002-1274-4598}{}

\author{Karthik Murali}{Simon Fraser University, Burnaby, Canada}{kmurali@sfu.ca}{https://orcid.org/0009-0003-3985-3609}{Part of the research work was done when the author was a student at Carleton University, Ottawa, Canada.}

\author{Michiel Smid}{Carleton University, Ottawa, Canada}{michiel@scs.carleton.ca}{https://orcid.org/0000-0003-3683-9054}{}

\authorrunning{A.Maheshwari, K.Murali and M.Smid} 

\Copyright{Anil Maheshwari, Karthik Murali and Michiel Smid} 

\ccsdesc[500]{Mathematics of computing~Graph Theory}

\keywords{CIST, Steiner Trees, Pendant Steiner Trees} 

\funding{This research was funded in part by NSERC.}

\nolinenumbers 

\usepackage{cleveref}
\usepackage{mdframed}
\usepackage{mathtools}
\DeclarePairedDelimiter{\floor}{\lfloor}{\rfloor}
\DeclarePairedDelimiter{\ceil}{\lceil}{\rceil}

\theoremstyle{plain}   

\usepackage{cite}
\usepackage{tcolorbox}

\begin{document}

\maketitle

\begin{abstract}
Spanning trees are fundamental for efficient communication in networks. For fault-tolerant communication, it is desirable to have multiple spanning trees to ensure resilience against failures of nodes and edges. To this end, various notions of disjoint or independent spanning trees have been studied, including edge-disjoint, node/edge-independent, and completely independent spanning trees. Alongside these, several Steiner variants have also been investigated, where the trees are required to span a designated subset of vertices called terminals. For instance, the study of edge-disjoint spanning trees has been extended to edge-disjoint Steiner trees; a stronger variant is the problem of internally disjoint Steiner trees, where any two Steiner trees intersect exactly in the terminals.

In this paper, we investigate the Steiner analogue of completely independent spanning trees, which we call \emph{completely independent Steiner trees}. A set of Steiner trees is completely independent if, for every pair of terminals $u,v$, the $(u,v)$-paths in all the Steiner trees are internally vertex-disjoint and edge-disjoint. This notion generalizes both completely independent spanning trees and internally disjoint Steiner trees. We provide a systematic study of completely independent Steiner trees from structural, algorithmic, and complexity-theoretic perspectives. In particular, we present several characterisations, connectivity bounds, algorithms, hardness results, and applications to special graph classes such as planar graphs and graphs of bounded treewidth. Along the way, we also introduce a directed variant of completely independent spanning trees via an equivalence with completely independent Steiner trees. 
\end{abstract}

\newpage
\setcounter{page}{1}

\section{Introduction}

Spanning trees are important for computer networks since they represent structures through which information can be efficiently broadcasted throughout the network. For fault-tolerant broadcast, it is ideal to have spanning trees that are independent or disjoint in some respect. There is a vast literature on the many different types of spanning trees that serve this purpose---examples include edge-disjoint spanning trees, node/edge-independent spanning trees, and completely independent spanning trees. Below, we give an overview of these different types of spanning trees.

A set of spanning trees is \textit{edge-disjoint} if no two trees share a common edge. Such trees are useful for network reliability, since any network having $k$ edge-disjoint spanning trees is resilient to the failure of up to $(k-1)$ edges. The problem of finding the maximum number of edge-disjoint spanning trees in a graph is called the \textit{spanning-tree packing problem}. A classical result due to Tutte and Nash-Williams states that if a graph is $2k$-edge-connected, then there exist $k$ edge-disjoint spanning trees \cite{Tutte1961,NashWilliams1961}. 

Another type of spanning trees, especially useful for problems related to single-source broadcasting, is independent spanning trees \cite{itai}. A set of spanning trees rooted at a vertex $r$ is \textit{node-independent} (resp. \textit{edge-independent}) if, for every vertex $u$, the $(u,r)$-paths in the all the trees are internally node-disjoint (resp. edge-disjoint). Every set of node-independent spanning trees is also edge-independent, but not conversely. And, it is not necessary that a set of node/edge-independent spanning trees also be edge-disjoint spanning trees. Notice that if a communication network has $k$ node-independent (resp. edge-independent) rooted spanning trees, then the failure of any $k-1$ nodes other than the root (resp. $k-1$ edges) does not disrupt communication with the root. It has been conjectured that if a graph is $k$-vertex-connected (resp. $k$-edge-connected), then there are $k$ node-independent (resp. edge-independent) spanning trees rooted at any arbitrary vertex \cite{chartrand,itai}. Despite almost four decades since these conjectures were posed, they remain unsolved, except for very small values of $k$ \cite{node_best,edge_best}. Independent spanning trees have been studied for various types of interconnection networks such as hypercubes, butterfly networks, pancake networks, bubble-sort networks, etc; see \cite{survey} for a detailed survey on this topic.

\textit{Completely independent spanning trees} was introduced by Hasunuma as a generalisation of node-independent spanning trees \cite{hasunuma}. A set of spanning trees is completely independent if, for \textit{every} pair of vertices $u,v$, the $(u,v)$-paths in the trees are both edge-disjoint and internally node-disjoint. In other words, a set of completely independent spanning trees is a set of node-independent spanning trees, where any vertex of the graph can play the role of the root. It was initially conjectured that every $2k$-vertex-connected graph has $k$ completely independent spanning trees \cite{maximal_planar}. But this was later shown to be false---for every $k \geq 2$, there is a $k$-vertex-connected graph for which there do not exist two completely independent spanning trees \cite{peterfalvi}.

A natural generalisation of all the above problems is to consider their Steiner versions, where we do not require trees to span all vertices of $G$, but only a select subset of them, called \textit{terminals}. A \textit{Steiner tree} in a graph is any tree that spans all the terminals. Apart from communication networks, Steiner trees have well-known applications in VLSI design and phylogeny \cite{Steiner_book}. The Steiner version of the spanning-tree packing problem is the \textit{Steiner-tree packing problem}, where the goal is to find a maximum set of edge-disjoint Steiner trees \cite{kriesell}. A stronger version of this problem requires that any two trees intersect exactly in the set of terminals; in other words, the edges and non-terminal vertices of any two trees are disjoint. Such a set of trees is called \textit{internally disjoint Steiner trees} \cite{rainbow_trees} or \textit{element disjoint Steiner trees} \cite{cheriyan_talg}. Let $R$ be a set of terminals in a graph $G$. Notice that when $|R| = 2$, a maximum set of internally disjoint (resp. edge-disjoint) Steiner trees is simply a maximum set of vertex-disjoint (resp. edge-disjoint) paths between the two vertices of $R$. When we consider minimising this among all possible terminals of size $2$ in the graph, then by Menger's theorem, this value is exactly the vertex connectivity (resp. edge connectivity) of the graph. Hence, the problem of finding the maximum number of internally disjoint (resp. edge-disjoint) Steiner trees, minimised over all possible choice of $k$ terminals, has been called the \textit{generalised $k$-vertex-connectivity of graphs} \cite{chartrand_general} or the \textit{$k$-tree-connectivity of graphs} \cite{Hager} (resp. \textit{generalised $k$-edge-connectivity of graphs} \cite{general_edge}). 
The interested reader may refer to the book \cite{tree_book} by Li and Mao for a detailed treatment on the generalised connectivity of graphs.

\subparagraph{Our Contribution.}

In this paper, we study the Steiner version of completely independent spanning trees. Given a graph $G$ and a set of terminals $R$, we call a set of Steiner trees completely independent if, for any pair of terminals $u,v \in R$, the $(u,v)$-paths in the trees are internally vertex-disjoint and edge-disjoint. Independently of our work, the concept of completely independent Steiner trees was also studied very recently by Yuan, Liu, Lin and Liu \cite{cisst}. While their focus is on the tree-connectivity version of the problem, mainly for complete graphs and complete bipartite graphs, our focus here is to provide a comprehensive study on the various structural characterisations, algorithms, hardness results, and applications to different graph classes. We start by formally introducing various definitions and notation in \Cref{sec: prelims}. Then, in \Cref{sec: characterisations}, we give three different ways of characterising completely independent Steiner trees---\Cref{subsec: structural characterisation} gives a characterisation based on structure, \Cref{subsec: connected r-dominating sets} gives a characterisation through connected dominating sets, and \Cref{subsec: cisa} gives another characterisation through directed completely independent spanning trees. As part of \Cref{subsec: cisa}, we also briefly study the concept of completely independent spanning trees in digraphs, a topic which, to the best of our knowledge, has not been studied before. Accompanying these characterisations, we also give applications to graph classes such as planar graphs and graphs of bounded treewidth. In \Cref{sec: connectivity}, we investigate the relation between vertex connectivity and completely independent Steiner trees. We show that if a graph $G$ is $10kr$-vertex-connected, then $G$ has $k$ completely independent Steiner trees for any choice of $r$ terminals in the graph. In \Cref{sec: algorithms}, we show that it is NP-complete to decide whether a graph $G$ has $k$ completely independent Steiner trees for a set of $r$ terminals, but polynomial-time solvable when both $r$ and $k$ are constants. We conclude in \Cref{sec: conclusion} by presenting potential avenues for future work.

\section{Preliminaries and Definitions}\label{sec: prelims}

All graphs considered in this paper are unweighted, connected and finite, but not necessarily simple---we allow parallel edges between the same pair of vertices, but no self-loops. Let $G = (V,E)$ be a graph. For any vertex $v \in V(G)$, we write $d_G(v)$ to denote the \emph{degree} of $v$ in $G$ (i.e., the number of edges incident with $v$), $N_G(v)$ to denote the set of vertices of $G$ adjacent to $v$, and $N_G[v]$ for the union $N_G(v) \cup \{v\}$. If $H$ is a subgraph of $G$, we write $H \subseteq G$. For any subset of vertices $W \subseteq V(G)$, the graph $G[W]$ \emph{induced by $W$} is the edge-maximal subgraph of $G$ on vertex set $W$. For any subset of edges $F \subseteq E(G)$, the graph $G[F]$ \emph{induced by $F$} is the subgraph of $G$ whose edge set is exactly $F$. If $P$ is a path in a graph $G$ with end vertices $u$ and $v$, we say that $P$ is a \emph{$(u,v)$-path}. Two paths $P_1 := P_1(u,v)$ and $P_2 := P_2(u,v)$ are \emph{internally vertex-disjoint} if $V(P_1) \cap V(P_2) = \{u,v\}$. If $T$ is a tree in $G$, and $\{u,v\} \subseteq V(T)$, then we use $T(u,v)$ to denote the unique $(u,v)$-path in $T$. A vertex $v \in V(T)$ is an \emph{internal node} of $T$ if $d_T(v) \geq 2$; else $v$ is a \emph{leaf} of $T$. We write $\mathrm{int}(T)$ to denote the set of all internal nodes of $T$ and $L(T)$ to denote its set of all leaves. A \emph{star} is a tree with exactly one internal node, and a \emph{double star} is a tree with exactly two internal nodes. We say that two graphs $G_1$ and $G_2$ are \emph{homeomorphic} if there exists a graph $G$ such that both $G_1$ and $G_2$ are subdivisions of $G$. For any positive integer $n$, we use the shorthand $[n]$ for the set $\{1,\dots,n\}$.

\begin{definition}[$R$-Steiner Tree]\label{def: steiner tree}
 Let $R \subseteq V(G)$ be a set of vertices, where $|R| \geq 2$. Any tree $T$ of $G$ where $R \subseteq V(T)$ and $L(T) \subseteq R$ is called an \emph{$R$-Steiner tree}.
\end{definition}

The vertices of $R$ in \Cref{def: steiner tree} are called \emph{terminals}, and the remaining vertices are called \emph{non-terminals}. Henceforth, whenever we speak of $R$-Steiner trees, we assume that a set $R$ of terminals is specified. An $R$-Steiner tree $T$ where $L(T) = R$ is called a \emph{pendant $R$-Steiner tree} \cite{Hager}. (The literature on minimum-weight Steiner trees in edge-weighted graphs uses the term \emph{full Steiner tree} \cite{LU200355,hwang1995steiner} instead of pendant Steiner tree.)  An $R$-Steiner tree where $\mathrm{int}(T) \cap R \neq \emptyset$ is called a \emph{non-pendant $R$-Steiner tree}.

\begin{definition}[$R$-Completely Independent Steiner Trees ($R$-CIST)]\label{def: independent steiner trees}
    An \textsf{$R$-CIST} of $G$ is a set of trees $\{T_1, \dots, T_k\}$ where, for every distinct $\{u,v\} \subseteq  R$, and all distinct $i,j \in [k]$, paths $T_i(u,v)$ and $T_j(u,v)$ are edge-disjoint and internally vertex-disjoint.
\end{definition}

An \textsf{$R$-CIST} whose each tree is a pendant $R$-Steiner tree is called a \emph{\textsf{Pendant} \textsf{$R$-CIST}}. On the other hand, an  \textsf{$R$-CIST} whose each tree is a non-pendant $R$-Steiner tree is called a \emph{\textsf{Non-pendant} \textsf{$R$-CIST}}. If $R = V(G)$, then an \textsf{$R$-CIST} of $G$ is a set of \emph{completely independent spanning trees}, or \textsf{CIST} for short \cite{hasunuma}.

\begin{definition}[$R$-Disjoint Steiner Trees (\textsf{$R$-DST})]\label{def: rdst}
    An  \textsf{$R$-DST} is a set of $R$-Steiner trees $\{T_1, \dots, T_k\}$ where $V(T_i) \cap V(T_j) = R$ and $E(T_i) \cap E(T_j) = \emptyset$ for all distinct $i,j \in [k]$.
\end{definition}

The concept of $R$-disjoint Steiner trees has also been studied under different names such as \emph{element-disjoint Steiner trees} \cite{cheriyan_talg} and \emph{internally vertex-disjoint Steiner trees} \cite{tree_book,Hager}. 
It is easy to construct a family of graphs where a maximum \textsf{$R$-CIST} has size 1, but there is no bound on a maximum \textsf{$R$-DST}. For instance, consider the infinite family of graphs $\{H_1, H_2, \dots,\}$, where $V(H_i) =\{u,x,v\}$, and $E(H_i)$ consists of two sets of $i$ parallel edges each, with one set of parallel edges between $u,x$, and the other between $x,v$. Clearly, if $R = \{u,x,v\}$, then every graph $H_i$ has a maximum \textsf{$R$-CIST} of size 1, but a maximum \textsf{$R$-DST} has size $i$. While an \textsf{$R$-DST} is not necessarily an \textsf{$R$-CIST}, every \textsf{$R$-CIST} is an \textsf{$R$-DST} (see \Cref{cor: dst} in \Cref{subsec: structural characterisation}). However, for graphs where every terminal vertex has degree at most 3, every \textsf{$R$-DST} is also an \textsf{$R$-CIST}, since no two $(u,v)$-paths, where $\{u,v\} \subseteq R$, can intersect at vertex $z \in R \setminus \{u,v\}$. 


\medskip
In \Cref{sec: characterisations}, we give various characterisations for \textsf{$R$-CIST} and their applications. For many of these, we will need to assume that $R$ is an independent set in the graph. In \Cref{proposition: subdivision}, we show that such an assumption is justifiable. 

\begin{proposition}\label{proposition: subdivision}
    Let $G$ be a graph, $R \subseteq V(G)$ be a set of terminals, and $G'$ be any subdivision of $G$. Then $G$ has an \textsf{$R$-CIST} $\{T_1, \dots, T_k\}$ if and only if $G'$ has an \textsf{$R$-CIST} $\{T_1', \dots, T_k'\}$, where $T_i'$ is the subdivision of $T_i$ in $G'$. 
\end{proposition}

\begin{proof}
For any $\{u,v\} \subseteq R$, and distinct $i,j \in [k]$, the paths $T_i(u,v)$ and $T_j(u,v)$ are edge-disjoint and internally vertex-disjoint if and only if paths $T_i'(u,v)$ and $T_j'(u,v)$ are also edge-disjoint and internally vertex-disjoint. Therefore, $\{T_1, \dots, T_k\}$ is an \textsf{$R$-CIST} of $G$ if and only if $\{T_1', \dots, T_k'\}$ is an \textsf{$R$-CIST} of $G'$.  
\end{proof}

\section{Characterisations of Completely Independent Steiner Trees}\label{sec: characterisations}

In this section, we give three different characterisations for \textsf{$R$-CIST}, and give different applications for each of them. 

\subsection{Characterisation by Structure}\label{subsec: structural characterisation}

The main result of this section is \Cref{theorem: disjoint internal nodes}, where we show that a set of trees is an \textsf{$R$-CIST} if and only if they are pairwise edge-disjoint and disjoint on their internal nodes. Independently of our work, \Cref{theorem: disjoint internal nodes} was also proved in \cite{cisst}. Nevertheless, we repeat the proof of \Cref{theorem: disjoint internal nodes}, since this result is foundational to our paper, and makes the paper self-contained.  \Cref{theorem: disjoint internal nodes} draws its inspiration primarily from Hasunuma's structural characterisation theorem for \textsf{CIST}, which we state below. 

\begin{theorem}[Theorem 2.1 in \cite{hasunuma}]\label{theorem: cist characterisation}
    A set of trees $\{T_1, \dots, T_k\}$ is a \textsf{CIST} in a graph $G$ if and only if $E(T_i) \cap E(T_j) = \emptyset$, and for any vertex $v \in V(G)$, there is at most one tree $T_i$ such that $v \in \mathrm{int}(T_i)$.
\end{theorem}

\begin{theorem}[See also \cite{cisst}]\label{theorem: disjoint internal nodes}
For any graph $G$, a set of trees $\{T_1, \dots, T_k\}$ is an \textsf{$R$-CIST} if and only if $E(T_i) \cap E(T_j) = \emptyset$ and $\mathrm{int}(T_i) \cap \mathrm{int}(T_j) = \emptyset$ for all distinct $i,j \in [k]$.
\end{theorem}

\begin{proof}
The reverse direction is easy to show: For any $\{u,v\} \subseteq R$, and any distinct $i,j \in [k]$, the paths $T_i(u,v)$ and $T_j(u,v)$ are edge disjoint (since $E(T_i) \cap E(T_j) = \emptyset$) and internally vertex-disjoint (since $\mathrm{int}(T_i) \cap \mathrm{int}(T_j) = \emptyset$). Therefore, $\{T_1, \dots, T_k\}$ is an \textsf{$R$-CIST}. 

To prove the forward direction, we use contradiction. By definition of \textsf{$R$-CIST}, we already have that $E(T_i) \cap E(T_j) = \emptyset$ for all distinct $i,j \in [k]$. Hence, we may assume that $\mathrm{int}(T_i) \cap \mathrm{int}(T_j) \neq \emptyset$, for some $i,j \in [k]$, and that there is a vertex $z \in \mathrm{int}(T_i) \cap \mathrm{int}(T_j)$. Let $\{u,v\} \subseteq L(T_i)$ and $\{w,x\} \subseteq L(T_j)$ be vertices such that $z$ is internal to paths $T_i(u,v)$ and $T_j(w,x)$. Now, root both trees $T_i$ and $T_j$ at the vertex $z$ (\Cref{fig: structure}). Vertices $w$ and $x$ belong to the same component of $T_i - \{z\}$; otherwise, since $z$ is an internal node of $T_j(w,x)$, paths $T_i(w,x)$ and $T_j(w,x)$ would both intersect at $z$, leading to a contradiction. As $z \in T_i(u,v)$ implies that $u$ and $v$ are in different components of $T_i - \{z\}$, we may assume, up to symmetry between $u$ and $v$, that $v$ is in a component different from $\{w,x\}$ (\Cref{fig:first}). By a symmetric reasoning for $T_j$, vertices $\{u,v\}$ belong to the same component of $T_j - \{z\}$, and up to symmetry between $w$ and $x$, we may assume that this component is different from the one containing $x$ (\Cref{fig:second}). One may now observe that $z \in T_i(v,x) \cap T_j(v,x)$---this contradicts the fact that $T_i$ and $T_j$ belong to an \textsf{$R$-CIST}.   
\end{proof}

\begin{figure}
\centering
\begin{subfigure}{0.48\textwidth}
    \centering
    \includegraphics[scale = 0.6, page = 1]{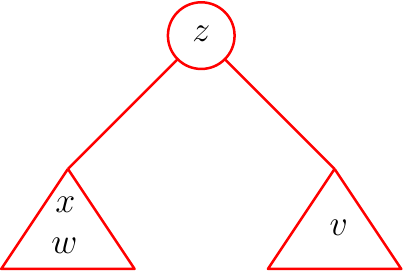}
    \caption{Tree $T_i$}
    \label{fig:first}
\end{subfigure}
\hfill
\begin{subfigure}{0.48\textwidth}
    \centering
    \includegraphics[scale = 0.6, page = 2]{Figures/structure.pdf}
    \caption{Tree $T_j$}
    \label{fig:second}
\end{subfigure}
\caption{An illustration for the proof of \Cref{theorem: disjoint internal nodes}}
\label{fig: structure}
\end{figure}

\cref{theorem: disjoint internal nodes} implies that any terminal vertex that is an internal node of some tree $T_i$ is a leaf of all other trees $T_j$. And, any non-terminal vertex that is an internal node of $T_i$ does not belong to any other $T_j$. In the special case $R = V(G)$, where an \textsf{$R$-CIST} is identical to a \textsf{CIST}, every vertex of $G$ is a leaf in all but except possibly one of the trees---this is precisely \Cref{theorem: cist characterisation}. 

\begin{corollary}\label{cor: dst}
Every \textsf{$R$-CIST} of a graph $G$ is also an \textsf{$R$-DST}.
\end{corollary}

\begin{proof}
Let $\{T_1, \dots, T_k\}$ be an \textsf{$R$-CIST} of $G$. By \Cref{theorem: disjoint internal nodes}, we have $E(T_i) \cap E(T_j) = \emptyset$ and $V(T_i) \cap V(T_j) = R$, for all distinct $i,j \in [k]$. Therefore, $\{T_1, \dots, T_k\}$ is also an \textsf{$R$-DST}.
\end{proof}

\begin{corollary}\label{obs: number of RCIST}
For any graph $G$, a maximum \textsf{Non-pendant $R$-CIST} of $G$ has size at most $|R|$. 
\end{corollary}

\begin{proof}
Any non-pendant $R$-Steiner tree has at least one terminal as an internal node. Since the internal nodes of any two trees in an \textsf{$R$-CIST} are pairwise disjoint, the size of any \textsf{Non-pendant} \textsf{$R$-CIST} in $G$ is at most $|R|$. 
\end{proof}

In \Cref{theorem: disjoint internal nodes}, although we do not assume $G$ to be a simple graph, the disjointness property of internal nodes implies that one can ignore all parallel edges in the graph, except those in induced subgraph $G[R]$. We prove this formally in \Cref{prop: simplification} after first defining a special type of simplified graph, which we denote by $\mathrm{simp}(G,R)$. 

\begin{definition}[Simplification of $G$ with respect to $R$ ($\mathrm{simp}(G,R)$)]
   For any graph $G$ with a set of terminals $R$, let $\mathrm{simp}(G,R)$ be the subgraph of $G$ obtained by the following process: For any pair of vertices $\{u,v\}$ in $G$ connected by more than a single edge:
   \begin{enumerate}
       \item If $\{u,v\} \subseteq R$ and $|R| \geq 3$, then delete all edges in $G[\{u,v\}]$ except two edges;
       \item Else if $u \notin R$ or $v \notin R$, then delete all edges in $G[\{u,v\}]$ except one edge 
   \end{enumerate}
\end{definition}

Note that if $R = \{u,v\}$, and $u,v$ are adjacent in $G$, we keep all the parallel edges between $u$ and $v$. Also, note that $\mathrm{simp}(G,R)$ is not a unique graph if parallel edges between vertices are labelled distinctly.

\begin{proposition}\label{prop: simplification}
A set of trees is an \textsf{$R$-CIST} of $G$ if and only if it is an \textsf{$R$-CIST} in some graph $\mathrm{simp}(G,R)$.
\end{proposition}

\begin{proof}
   The reverse direction is trivially true, since $\mathrm{simp}(G,R) \subseteq G$, so we only show the forward direction. Let $\{T_1, \dots, T_k\}$ be an \textsf{$R$-CIST} of $G$. We will show that there exists a graph $\mathrm{simp}(G,R)$ for which $\{T_1, \dots, T_k\}$ is also an \textsf{$R$-CIST}. Let $u,v$ be a pair of adjacent vertices of $G$. Suppose that there exist trees $T_i,T_j$ such that $e^1_{uv} \in E(T_i)$ and $e^2_{uv} \in E(T_j)$, where $e^1_{uv}$ and $e^2_{uv}$ are two distinct parallel edges with ends $u,v$. If $R = \{u,v\}$, then all parallel edges between $u,v$ belong to $\mathrm{simp}(G,R)$, so we may assume that $R \neq \{u,v\}$. In this case, we will show that $\{u,v\} \subseteq R$, and that no other edge of $G[\{u,v\}]$ belongs to the \textsf{$R$-CIST}. Since $\mathrm{int}(T_i) \cap \mathrm{int}(T_j) = \emptyset$, up to symmetry between $u$ and $v$, it must be that $u \in \mathrm{int}(T_i) \cap L(T_j)$ and $v \in \mathrm{int}(T_j) \cap L(T_i)$. Since $u$ and $v$ are both leaves of Steiner trees, it follows that $\{u,v\} \subseteq R$. Furthermore, since $u$ and $v$ are internal nodes of two trees, no other edge $e^3_{uv}$ between $u$ and $v$ can be an edge of a tree in $\{T_1, \dots, T_k\}$. Therefore, in the construction of $\mathrm{simp}(G,R)$, one can delete all edges in $G[\{u,v\}]$ except $e^1_{uv}$ and $e^2_{uv}$. By repeating this for all adjacent pairs of vertices $u,v$, we can construct a graph $\mathrm{simp}(G,R)$ for which $\{T_1, \dots, T_k\}$ is also an \textsf{$R$-CIST}.
\end{proof}

\begin{corollary}
If $R$ is an independent set of terminals in a graph $G$, then a set of trees is an \textsf{$R$-CIST} of $G$ if and only if it is also an \textsf{$R$-CIST} of a simple graph underlying $G$.
\end{corollary}

Since we do not restrict ourselves to simple graphs in this paper, it will be useful to define \emph{$R$-parallel trees}, which are trees with the same vertex set and same edge set, modulo the choice of parallel edges between vertices in $G[R]$.

\begin{definition}[$R$-parallel trees]
In any graph $G$, two trees $T_1$ and $T_2$ of $G$ are \emph{$R$-parallel} if the following condition holds: For every edge $e \in E(G)$:
  \begin{enumerate}
      \item If $e \notin E(G[R])$, then $e \in E(T_1)$ if and only if $e \in E(T_2)$

      \item If $e \in E(G[R])$, then $e \in E(T_1)$ if and only if $e' \in E(T_2)$, where $e' = e$, or $e'$ is an edge parallel to $e$.
  \end{enumerate}
\end{definition}

The presence of parallel edges in $G[R]$ is beneficial in at least two ways. Firstly, the size of a maximum \textsf{$R$-CIST} in a graph with parallel edges may potentially be larger than one without parallel edges. For instance, the size of a maximum \textsf{$R$-CIST} in a simple complete graph $G$ is $|V(G)| - \ceil{|R|/2}$ \cite{cisst}. However, if $G$ is a complete graph with parallel edges between every pair of terminals, then the size of a maximum \textsf{$R$-CIST} is equal to $|V(G)|$---this is realised by the set of stars, each having a distinct vertex of $G$ as an internal node. The second way in which parallel edges in $G[R]$ may be beneficial is that it can make computing the size of a maximum \textsf{$R$-CIST} easier. This is by virtue of the fact that one can forgo the requirement of edge-disjointness for \textsf{$R$-CIST}. For this purpose, we introduce the following definition, and elaborate further through \Cref{prop: idst}.

\begin{definition}[$R$-Interior Disjoint Steiner Trees ($R$-IDST)]
    A set of trees $\{T_1,\dots,T_k\}$ in a simple graph is an \textsf{$R$-IDST} if $\mathrm{int}(T_i) \cap \mathrm{int}(T_j) = \emptyset$, for all distinct $i,j \in [k]$.
\end{definition}

\begin{proposition}\label{prop: idst}
  Let $G$ be a graph with a set of terminals $R$, where $|R| \geq 3$, and every edge in $G[R]$ has exactly one parallel edge. A set of trees $\{T_1, \dots, T_k\}$ is an \textsf{$R$-IDST} of a simple graph underlying $G$ if and only if $\{T_1', \dots, T_k'\}$ is an \textsf{$R$-CIST} of $G$, where $\{T_i,T_i'\}$ are $R$-parallel in $G$ for all $i \in [k]$.
\end{proposition}

\begin{proof}
 The reverse direction is trivially true since every \textsf{$R$-CIST} is also an \textsf{$R$-IDST}, so we only prove the forward direction. Let $\{T_1, \dots, T_k\}$ be an \textsf{$R$-IDST} in a simple graph underlying $G$. If $\{T_1, \dots, T_k\}$ is not an \textsf{$R$-CIST} of $G$, then there exists an edge $e \in E(T_i) \cap E(T_j)$, for some distinct $i,j \in [k]$. Let $u,v$ be the endpoints of the edge $e$. As $|R| \geq 3$, we have $\mathrm{int}(T_i) \neq \emptyset$ for all $i \in [k]$. Since $\mathrm{int}(T_i) \cap \mathrm{int}(T_j) = \emptyset$, it must be that $u \in \mathrm{int}(T_i) \cap L(T_j)$ and $v \in \mathrm{int}(T_j) \cap L(T_i)$. This implies that $\{u,v\} \subseteq R$, and $e \notin E(T_h)$, for any $h \notin \{i,j\}$. Since, $\{u,v\} \subseteq R$, there exists an edge $e'$ parallel to $e$ in $G$, and we can replace $T_j$ by the tree $T_j - e + e'$.  Repeating this process for all edges $e$, we can make the trees edge-disjoint and get an \textsf{$R$-CIST} $\{T_1', \dots, T_k'\}$ in $G$, where $\{T_i,T_i'\}$ are $R$-parallel in $G$.
\end{proof}

\begin{corollary}\label{Cor: ridst to rcsit}
    If $R$ is an independent set of terminals in a graph $G$, then $\{T_1, \dots, T_k\}$ is an \textsf{$R$-CIST} of $G$ if and only if it is also an \textsf{$R$-IDST} of $G$.
\end{corollary}

To demonstrate that the presence of parallel edges in $G[R]$ can simplify computing the size of a maximum \textsf{$R$-CIST}, we take complete bipartite graphs as an example. In \cite{cisst}, the authors consider the problem of computing a maximum \textsf{$R$-CIST} for simple complete-bipartite graphs. This is more challenging than computing a maximum \textsf{$R$-IDST}, which we do in \Cref{theorem: complete bipartite graphs}. Together with \Cref{prop: idst}, this will imply simple bounds for a maximum \textsf{$R$-CIST} in complete bipartite multi-graphs.

\begin{theorem}\label{theorem: complete bipartite graphs}
  Let $G$ be a simple complete bipartite graph with partite sets $(A,B)$, and $R \subseteq V(G)$ be a set of terminals.
  \begin{enumerate}
      \item If $R \subseteq A$, then a maximum \textsf{$R$-IDST} of $G$ has size exactly equal to $|B|$; (The case where $R \subseteq B$ is symmetric) 
      \item Else, if $|R \cap A| = 1$ and $|R \cap B| > 0$, then a maximum \textsf{$R$-IDST} of $G$ has size exactly equal to $\min(|A|, |B|+1)$; (The case where $|R \cap B| = 1$ is symmetric)
      \item Else, if $|R \cap A| > 1$ and $|R \cap B| > 1$, then a maximum \textsf{$R$-IDST} of $G$ has size exactly equal to $\min(|A|, |B|)$
  \end{enumerate}
\end{theorem}

\begin{proof}
We prove Case (1) first, then Case (3), and prove Case (2) last. If $R \subseteq A$, then any \textsf{$R$-IDST} of $G$ contains one internal node from $B$. Therefore, the size of a maximum \textsf{$R$-IDST} is at most $|B|$. Indeed, one can construct such a set of trees by considering $|B|$ stars, each having a distinct vertex of $B$ as an internal node. 

For Case (3), we first show that there exists an \textsf{$R$-IDST} of size $\min(|A|, |B|)$. For this, we simply construct a maximum matching in $G$, and extend each edge $uv$ of the matching to an $R$-Steiner tree by attaching each terminal in $R \setminus \{u,v\}$ as a leaf. This gives us \textsf{$R$-IDST} of size $\min(|A|, |B|)$. Next, we show that any \textsf{$R$-IDST} $\mathcal{T}$ has size at most $\min\{|A|,|B|\}$. Since $|R \cap A| > 1$ and $|R \cap B| > 1$, every tree of $\mathcal{T}$ must have at least one internal node from $A$ and one from $B$. Since the interior nodes of any two trees in $\mathcal{T}$ are disjoint, $|\mathcal{T}| \leq \min(|A|, |B|)$.

For Case (2), we can construct an \textsf{$R$-IDST} of size $\min(|A|, |B|+1)$ as follows. We start with a maximum matching in the graph $G- (R\cap A)$, and extend each edge $uv$ of the matching to an $R$-Steiner tree by attaching each terminal in $R \setminus \{u,v\}$ as a leaf. This gives us a collection of trees of size $\min\{|A|-1,|B|\}$. To this collection, we add the star with $R \cap A$ as the interior node and $R \cap B$ as its leaves. (If $|R \cap B| = 1$, then the tree is not a star, but simply an edge.) With this, the size of the collection of trees becomes $\min(|A|, |B|+1)$, and since the interior nodes of any two trees are disjoint, this is an \textsf{$R$-IDST}.

We now show that every \textsf{$R$-IDST} $\mathcal{T}$ has size at most $\min(|A|, |B|+1)$. As with Case (3), if every tree of $\mathcal{T}$ has one internal node from $A$ and one from $B$, then $|\mathcal{T}| \leq \min(|A|, |B|) \leq \min(|A|, |B|+1)$. So, we assume that there exists a tree $T$ where all vertices of $A \cap R$ or $B \cap R$ occur as leaves. 
\begin{itemize}
    \item Case (i): If $|B \cap R| = 1$, then $T$ is the edge connecting $A \cap R$ and $B \cap R$. This implies that every tree in $\mathcal{T}\setminus \{T\}$ has one internal node from $A \setminus R$ and one internal node from $B \setminus R$. Therefore, $|\mathcal{T}| \leq  \min\{|A|-1,|B|-1\}+1 \leq \min(|A|,|B|+1)$.

    \item Case (ii): If $|B \cap R| > 1$, then $T$ is a star with the vertex of $A \cap R$ as its internal node. This implies that every tree in $\mathcal{T}\setminus \{T\}$ has one internal node from $A \setminus R$ and one internal node from $B$. Therefore, $|\mathcal{T}| \leq  \min\{|A|-1,|B|\}+1 = \min(|A|,|B|+1)$.\qedhere
\end{itemize}
\end{proof}

\subsection{Characterisation by Connected \texorpdfstring{$R$}{R}-Dominating Sets}\label{subsec: connected r-dominating sets}

In this section, we give a characterisation of \textsf{$R$-CIST} that is inspired by the \textsf{CIST}-partition theorem given by Araki \cite{araki_dirac}, stated below.  

\begin{theorem}[Theorem 2.3 in \cite{araki_dirac}]\label{theorem: araki-cist-partition}
   An undirected graph $G$ has a \textsf{CIST} of size $k$ if and only if there exists a partition $\{V_1, \dots, V_k\}$ of $V(G)$ such that: 
    \begin{enumerate}
        \item $G[V_i]$ is connected for all $i \in [k]$; and
        \item For all distinct $i,j \in [k]$, no component of the bipartite graph $B(V_i,V_j,G)$ is a tree  
    \end{enumerate}
\end{theorem}

Here, $B(V_i,V_j,G)$ is the bipartite subgraph of $G$ consisting of all edges with one end in $V_i$ and the other end in $V_j$.

\medskip
Our characterisation of \textsf{$R$-CIST} is given in \Cref{theorem: CIST partition}. This theorem requires that $R$ be an independent set of terminals; this assumption is justified by \Cref{proposition: subdivision}. Before we proceed further, we will require the definition of connected $R$-dominating sets.  

\begin{definition}[Connected $R$-Dominating Set]
 Let $R$ be a set of terminals. A set of vertices $Q$ is called a connected $R$-dominating set if $G[Q]$ is connected and every vertex of $R$ either belongs to $Q$ or is adjacent to some vertex of $Q$.
\end{definition}

\begin{theorem}\label{theorem: CIST partition}
    Let $R$ be an independent set of terminals in a graph $G$. There exists an \textsf{$R$-CIST} of size $k$ in $G$ if and only if there exists disjoint subsets $\{V_1, \dots, V_k\}$ of $V(G)$ such that $V_i$ is a connected $R$-dominating set for all $i \in [k]$. 
\end{theorem}

\begin{proof}
    Suppose that $G$ has an \textsf{$R$-CIST} $\{T_1, \dots, T_k\}$. For all $i \in [k]$, let $V_i := \mathrm{int}(T_i)$. Then, by \cref{theorem: disjoint internal nodes}, $V_i \cap V_j = \emptyset$ for all $i \neq j$, and since $T_i$ is an $R$-Steiner tree, $\mathrm{int}(T_i) = V_i$ is a connected $R$-dominating set. This proves one direction. For the other direction, suppose that there exist disjoint subsets $\{V_1, \dots, V_k\}$ such that $V_i$ is a connected $R$-dominating set for all $i \in [k]$. Then, we build a spanning tree for each $G[V_i]$ and add every vertex of $R \setminus V_i$ as a leaf of the spanning tree; let this tree be $T_i$. By construction, $\{T_1, \dots, T_k\}$ is an \textsf{$R$-IDST} of $G$, and therefore, also an \textsf{$R$-CIST} of $G$ (\Cref{Cor: ridst to rcsit}).
\end{proof}

\begin{corollary}\label{cor: connected r-dominating sets}
 Let $R$ be an independent set of terminals in a graph $G$. If $G$ has an \textsf{$R$-CIST} consisting of $p$ pendant $R$-Steiner trees and $q$ non-pendant $R$-Steiner trees, then $G$ has disjoint connected $R$-dominating sets $\{V_1, \dots, V_p, W_1, \dots, W_q\}$ such that $V_i \cap R = \emptyset$ for all $i \in [p]$ and $W_i \cap R \neq \emptyset$ for all $i \in [q]$.    
\end{corollary}

In \Cref{thm: minimal vertex cut}, we give a simple application of \Cref{theorem: CIST partition}.

\begin{proposition}\label{thm: minimal vertex cut}
    If $R$ is a minimal vertex cut of $G$ and $G - R$ has $k$ components, then $G$ has a \textsf{Pendant} \textsf{$R$-CIST} of size $k$.
\end{proposition}

\begin{proof}
Since $R$ is a minimal vertex cut, every vertex of $R$ has a neighbour in each component of $G-R$. Therefore, each component of $G-R$ is a connected $R$-dominating set.
\end{proof}

Since \Cref{theorem: CIST partition} is stated for graphs with an independent set of terminals, we will frequently consider auxiliary graphs obtained by subdividing edges between terminals to make it an independent set. In view of this, it will be useful to define the graph $\tilde{G}_R$.

\begin{definition}[$\tilde{G}_R$]
    For any graph $G$ and a set of terminals $R \subseteq V(G)$, we define $\tilde{G}_R$ to be the graph homeomorphic to $G$ obtained by subdividing every edge of $G[R]$ exactly once.  
\end{definition}

\subsubsection{Planar Graphs}

\Cref{theorem: CIST partition} can be used to show that planar graphs do not have \textsf{$R$-CIST} of large sizes. 

\begin{theorem}\label{theorem: planar graphs}
    For any planar graph $G$ and a set of terminals $R$:
    \begin{enumerate}       
        \item If $|R| = 3$, then any \textsf{$R$-CIST} of $G$ has size at most 5. 

        \item If $|R| \geq 4$, then any \textsf{$R$-CIST} of $G$ has size at most 4.
    \end{enumerate}
\end{theorem}

\begin{proof}
    We first prove (1). For the sake of contradiction, suppose that $G$ an \textsf{$R$-CIST} of size at least 6. From \Cref{proposition: subdivision}, $\tilde{G}_R$ also has an \textsf{$R$-CIST} of size $k \geq 6$. From \Cref{theorem: CIST partition}, we know that there exist $k$ disjoint connected $R$-dominating sets $\{V_1,\dots,V_k\}$. Of these sets, there exist 3 sets, say $\{V_1,V_2,V_3\}$, such that $R \subseteq V_1 \cup V_2 \cup V_3$. Since $\{V_4,V_5,V_6\}$ dominate $R$, and $|R| = 3$, we get a $K_{3,3}$-minor in $\tilde{G}_R$---this contradicts the planarity of $\tilde{G}_R$.

    Next, we prove (2). For the sake of contradiction, suppose that there exists an \textsf{$R$-CIST} of size $k \geq 5$. From \Cref{proposition: subdivision}, $\tilde{G}_R$ also has an \textsf{$R$-CIST} of size $k \geq 5$, and from \Cref{theorem: CIST partition}, we know that $\tilde{G}_R$ has $k$ disjoint connected $R$-dominating sets $\{V_1,\dots,V_k\}$. First, consider the case where each $V_i$ contains at most one terminal vertex. Let $\{V_1, V_2, V_3, V_4, V_5\}$ be sets such that $|V_i \cap R| = 1$ for $i \in [4]$, and $V_5 \cap R = \emptyset$. By contracting each $V_i$ into a single vertex, we get a $K_5$-minor. Next, consider the case where there exist two sets, say $V_1$ and $V_2$, that together contain three terminal vertices, say $\{r_1,r_2,r_3\}$. Since all of $\{V_3,V_4,V_5\}$ all dominate $\{r_1,r_2,r_3\}$, we get a $K_{3,3}$-minor. In either case, this contradicts the planarity of $\tilde{G}_R$.
\end{proof}

It is easy to construct examples which demonstrate that the bounds in \Cref{theorem: planar graphs} cannot be improved; see \cref{fig: planar graphs} for example.

\begin{figure}
\centering

\begin{subfigure}[b]{0.48\textwidth}
    \centering
    \includegraphics[scale = 0.75, page = 1]{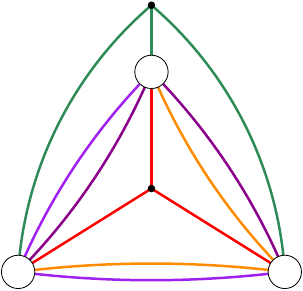}
\end{subfigure}
\hfill
\begin{subfigure}[b]{0.48\textwidth}
    \centering
    \includegraphics[scale = 0.75, page = 2]{Figures/planar_example.pdf}
\end{subfigure}

\caption{Examples for \Cref{theorem: planar graphs} that demonstrate its tightness. The terminals are shown with white disks. On the left is a planar graph with $|R| = 3$ and an \textsf{$R$-CIST} of size 5, where each tree is a star (shown in distinct colors). On the right is another planar graph with $|R|=4$ and an \textsf{$R$-CIST} of size 3, where each tree is again a star (shown in distinct colors).}
\label{fig: planar graphs}
\end{figure}

\subsubsection{Graphs of Bounded Treewidth}

First, we review the definitions of tree decomposition and treewidth. A \emph{tree decomposition} of a graph $G$ is a tree $\mathcal{T}$ whose nodes are associated with subsets of $V(G)$ called \emph{bags} such that: (1) every vertex $v$ of $G$ belongs to at least one bag of $\mathcal{T}$; (2) for every edge $uv$ of $G$, there exists a bag $Y\in V(\mathcal{T})$ with $u,v\in Y$; and (3) for every vertex $v$ of $G$, the set of bags containing $v$ forms a connected subtree of $\mathcal{T}$. 
The \emph{width} of a tree decomposition is $\max_{Y\in V(\mathcal{T})} |Y|-1$,
and the \emph{treewidth} of a graph $G$, denoted by $tw(G)$, is the smallest width of a tree
decomposition of $G$. For example, the treewidth of the complete graph $K_n$ is $n-1$, and the treewidth of the complete bipartite graph $K_{p,q}$ is $\min\{p,q\}$. Let $G$ and $H$ be two graphs. If $G$ is a minor of $H$, then $tw(G) \leq tw(H)$. As a consequence, if $G$ contains $K_n$ or $K_{p,q}$ as a minor, then $tw(G) \geq n-1$ or $tw(G) \geq \min\{p,q\}$, respectively. If a graph $H$ is a subdivision of a graph $G$, then $tw(H) \leq tw(G)$; since we also have $tw(G) \leq tw(H)$, due to $G$ being a minor of $H$, we have $tw(H) = tw(G)$. In \Cref{theorem: bdd tw pendant,theorem: bounded treewidth,cor: bounded tw bounded rcist}, we use these facts to show that graphs with bounded treewidth have only bounded-sized \textsf{$R$-CIST}.

\begin{theorem}\label{theorem: bdd tw pendant}
If $G$ is a graph of treewidth at most $w$ and $|R| \geq w+1$, then any \textsf{Pendant} \textsf{$R$-CIST} of $G$ has size at most $w$.
\end{theorem}

\begin{proof}
    Suppose, for the sake of contradiction, that $G$ has a \textsf{Pendant} \textsf{$R$-CIST} of size $p \geq w+1$. Consider the graph $\tilde{G}_R$, and note that $tw(\tilde{G}_R) = tw(G)$, since $\tilde{G}_R$ is a subdivision of $G$. By \Cref{proposition: subdivision}, $\tilde{G}_R$ also has a \textsf{Pendant} \textsf{$R$-CIST} of size $p \geq w+1$. By \Cref{cor: connected r-dominating sets}, there exist connected $R$-dominating sets $\{V_1, \dots, V_p\}$ in $\tilde{G}_R$ such that $V_i \cap R = \emptyset$ for all $i \in [p]$. By contracting each $G[V_i]$ to a single vertex, we find a $K_{p,r}$ minor in $\tilde{G}_R$, where $r = |R|$. This implies that $tw(G) \geq \min\{p,r\} \geq w+1$---this is a contradiction. 
\end{proof}

In \Cref{theorem: bdd tw pendant}, it is not possible to relax the number of terminals to less than $w+1$, and at the same time, also bound the size of a maximum \textsf{Pendant $R$-CIST}. For example, let $G$ be the complete bipartite graph $K_{r,t}$, and let the part containing the $r$ vertices be the set of terminals. If $r \leq w$, then $tw(G) \leq w$, but $G$ has a \textsf{Pendant} \textsf{$R$-CIST} of size $t$, where each tree is a star with a distinct internal node from the set of $t$ vertices. Since $t$ can be arbitrarily large, the size of a maximum \textsf{Pendant} \textsf{$R$-CIST} is not bounded.

\begin{theorem}\label{theorem: bounded treewidth}
If $G$ is a graph of treewidth at most $w$, then any \textsf{Non-pendant} \textsf{$R$-CIST} of $G$ has size at most $w+1$. 
\end{theorem}

\begin{proof}
Suppose, for contradiction, that $G$ has a \textsf{Non-pendant} \textsf{$R$-CIST} of size more than $w+1$. Then, $\tilde{G}_R$ also has a \textsf{Non-pendant} \textsf{$R$-CIST} of size more than $w+1$ (\Cref{proposition: subdivision}). By \Cref{cor: connected r-dominating sets}, there exist disjoint connected $R$-dominating sets $\{V_1, \dots, V_{w+2}\}$ in $\tilde{G}_R$ such that $V_i \cap R \neq \emptyset$. By contracting each $G[V_i]$ to a single vertex, we get the complete graph $K_{w+2}$ as a minor in $\tilde{G}_R$---this is a contradiction since $tw(\tilde{G}_R) = tw(G) \leq w$. 
\end{proof}

\begin{corollary}\label{cor: bounded tw bounded rcist}
  If $G$ is a graph of treewidth at most $w$ and $|R| \geq w+1$, then any \textsf{$R$-CIST} of $G$ has size at most $2w$.  
\end{corollary}

\begin{proof}
 From \Cref{theorem: bdd tw pendant,theorem: bounded treewidth}, it is clear that any \textsf{$R$-CIST} of $G$ has size at most $2w+1$. However, if $G$ were to have an \textsf{$R$-CIST} of size $2w+1$, then there must be exactly $w+1$ non-pendant $R$-Steiner trees, and $w$ pendant $R$-Steiner trees. By \Cref{proposition: subdivision}, the same holds true for $\tilde{G}_R$, and by \Cref{cor: connected r-dominating sets}, there exist disjoint connected $R$-dominating sets $\{V_1, W_1\dots, W_{w+1}\}$ such that $V_1 \cap R = \emptyset$ and $W_i \cap R \neq \emptyset$ for all $i \in [w+1]$. Contracting each of these into a single vertex gives us a $K_{w+2}$ minor in $\tilde{G}_R$, which is a contradiction since $tw(\tilde{G}_R) = tw(G) \leq w$. 
\end{proof}

\subsection{Characterisation by \textsf{CIST} in Directed Graphs}\label{subsec: cisa}

We now give a third characterisation that links \textsf{$R$-CIST} in a graph $G$ to \textsf{CIST} in a minor of $G$. The minor that we consider here will be a directed graph (digraph). Before proceeding further, we give the notation used for digraphs in this paper. We use the term \emph{arc} instead of edge for digraphs.  If a digraph contains an arc from vertex $u$ to vertex $v$, then the arc is written as $(u,v)$. (For undirected graphs, the edge is simply written as $uv$.) The set of all arcs in a digraph $G$ will be denoted by $A(G)$. For any vertex $v \in V(G)$, the \emph{in-degree} (number of in-coming arcs) and \emph{out-degree} (number of out-going arcs) of $v$ are denoted by $d_G^-(v)$ and $d_G^+(v)$, respectively. We use $\delta^-(G) := \min_{v} d^-_G(v)$ and $\delta^+(G) := \min_{v} d^+_G(v)$ to denote the \emph{minimum in-degree} and \emph{minimum out-degree} of $G$, respectively. The \emph{minimum semi-degree} of $G$ is $\delta^0(G) := \min\{\delta^-(G),\delta^+(G)\}$. 

\begin{definition}[Arborescence]\label{def: arbor}
    A tree $T$ in a digraph $G$ is called an \emph{arborescence} if, for all vertices $v \in T$, except for one vertex called the \emph{root}, $d^-_T(v) = 1$. (The root vertex has in-degree 0.) 
\end{definition}

Using \Cref{def: arbor}, one can define an analogue of \textsf{CIST} for digraphs, as shown below.

\begin{definition}[Completely Independent Spanning Arborescences (\textsf{CISA})]\label{def: cisa}
A set of spanning arborescences $\{T_1, \dots, T_k\}$ is \textsf{CISA} in a digraph $G$ if $\mathrm{int}(T_i) \cap \mathrm{int}(T_j) = \emptyset$ and $A(T_i) \cap A(T_j) = \emptyset$ for all distinct $i,j \in [k]$.
\end{definition}

\Cref{theorem: CISA fwd,thm: CISA rcist 2,theorem: CISA bwd} give the connection between \textsf{$R$-CIST} in $G$ and \textsf{CISA} in a minor of $G$ whose edges are all directed; the description of this minor is given below. 

\begin{definition}[Directed $R$-minor]\label{def: directed minor}
Let $\mathcal{P} := \{V_1, \dots, V_{|R|}\}$ be a partition of $V(G)$ such that $|V_i \cap R| = 1$ and $G[V_i]$ is connected for all $i \in [|R|]$. The \emph{directed $R$-minor of $G$} induced by $\mathcal{P}$ is the digraph with vertex set $\mathcal{P}$, and in which there is an arc $(V_i,V_j)$ if and only if there is an edge $ij \in E(G)$, where $i \in V_i$ and $j \in V_j \cap R$. (In particular, if two terminals are adjacent in $G$, then there is a directed 2-cycle in the directed $R$-minor.)
\end{definition}

\begin{theorem}\label{theorem: CISA fwd}
If $G$ has a \textsf{Pendant} \textsf{$R$-CIST} or a \textsf{Non-pendant $R$-CIST} of size $k \leq |R|$, then there exists a directed $R$-minor of $G$ that has a \textsf{CISA} of size $k$.
\end{theorem}

\begin{proof}
Let $\mathcal{T} := \{T_1, \dots, T_k\}$ be either a \textsf{Pendant $R$-CIST} or a \textsf{Non-pendant $R$-CIST} of $G$. We contract different connected subgraphs of $G$ into the terminal vertices, and the directed $R$-minor of $G$ thus induced will have a \textsf{CISA} of size $k$. Let $V(\mathcal{T}) := \bigcup_{i\in [k]} V(T_i)$. At the outset, we contract each component of $G - V(\mathcal{T})$ into some vertex of $V(\mathcal{T})$. Therafter, we deal with the cases when $\mathcal{T}$ is a \textsf{Pendant} \textsf{$R$-CIST} and a \textsf{Non-pendant} \textsf{$R$-CIST} separately.  

If $\mathcal{T}$ is a \textsf{Pendant} \textsf{$R$-CIST}, then we contract each $\mathrm{int}(T_i)$ into a distinct terminal vertex $r_i$; this is feasible since $|\mathcal{T}| \leq |R|$. This gives us a directed $R$-minor with a \textsf{CISA} $\{T_1', \dots, T_k'\}$, where each tree $T_i'$ is a spanning star-arborescence rooted at $r_i$. 

If $\mathcal{T}$ is a \textsf{Non-pendant $R$-CIST}, we first root each tree $T_i$ at an internal terminal node, say $r_i$. For each terminal $v \in R$ and $T_i \in \mathcal{T}$, let $T_i(v)$ be the maximal sub-tree of $T_i$ rooted at $v$ such that $V(T_i(v)) \cap R = v$. Since $\mathrm{int}(T_i) \cap \mathrm{int}(T_j) = \emptyset$ for all distinct $i,j \in [k]$ (\Cref{theorem: disjoint internal nodes}), $T_i(v) \cap T_j(v) = v$ for all distinct $i,j \in [k]$. For each terminal $v \in R$ and tree $T_i \in \mathcal{T}$, we contract $V(T_i(v))$ into $v$. This gives us a \textsf{CISA} $\{T_1', \dots, T_k'\}$ in the corresponding directed $R$-minor, where each $T_i'$ is a spanning arborescence rooted at $r_i$.
\end{proof}

\begin{figure}
\centering

\begin{subfigure}[b]{0.48\textwidth}
    \centering
    \includegraphics[scale = 1, page = 1]{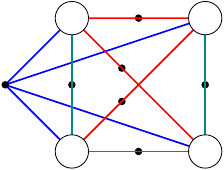}
\end{subfigure}
\hfill
\begin{subfigure}[b]{0.48\textwidth}
    \centering
    \includegraphics[scale = 1, page = 2]{Figures/K4_subdivision.pdf}
\end{subfigure}
\caption{On the left is a graph that is a subdivision of $K_5$. The terminal vertices are white disks, and the non-terminals are black dots. This graph has an \textsf{$R$-CIST} of size 3, consisting of one pendant and two non-pendant $R$-Steiner trees (trees are highlighted with different colours). Any directed $R$-minor of this graph is a digraph whose underlying simple graph is $K_4$ (see right). Moreover, this digraph has a single vertex incident with three pairs of parallel edges, with one edge from each pair directed away from the vertex. One can easily argue that there does not exist a \textsf{CISA} of size 3 in such a graph.}
\label{fig: k4 subdivision}
\end{figure}

In \Cref{theorem: CISA fwd}, the requirement that the sizes of the \textsf{Pendant $R$-CIST} and \textsf{Non-pendant $R$-CIST} be at most $|R|$ is natural since the maximum size of a $\textsf{CISA}$ in any directed $R$-minor of $G$ is at most $|R|$, unless $|R| = 2$. Although one cannot generalise \Cref{theorem: CISA fwd} to any $\textsf{$R$-CIST}$ of size $\leq |R|$ without further assumptions (see \Cref{fig: k4 subdivision}), the statement holds true for all \textsf{$R$-CIST} of size 2, as we show in \Cref{thm: CISA rcist 2}.

\begin{theorem}\label{thm: CISA rcist 2}
If a graph $G$ has an \textsf{$R$-CIST} of size 2, then there is a directed $R$-minor of $G$ that has a \textsf{CISA} of size 2.   
\end{theorem}

\begin{proof}
If the \textsf{$R$-CIST} $\mathcal{T}$ is a \textsf{Pendant $R$-CIST} or a \textsf{Non-pendant $R$-CIST}, then we are already done by \Cref{theorem: CISA fwd}. So, we assume that $\mathcal{T} : = \{T_p,T_{np}\}$, where $T_p$ a pendant $R$-Steiner tree and $T_{np}$ is a non-pendant $R$-Steiner tree. As in the proof of \Cref{theorem: CISA fwd}, we contract each component of $G - V(\mathcal{T})$ into some vertex of $V(\mathcal{T})$. Then, we contract all the interior nodes of $T_p$ into a terminal vertex that is a leaf of $T_{np}$. Then, as in the proof of \Cref{theorem: CISA fwd}, we consider for each terminal $v \in R$, the maximal sub-tree $T_{np}(v)$, and contract $V(T_{np}(v))$ into $v$. This will give us a \textsf{CISA} of size 2.   
\end{proof}

In \Cref{theorem: CISA bwd}, we give the converse of \Cref{theorem: CISA fwd}, but without qualifying the \textsf{$R$-CIST} as pendant or non-pendant. 

\begin{theorem}\label{theorem: CISA bwd}
If there exists a directed $R$-minor of $G$ that has a \textsf{CISA} of size $k$, then $G$ has an \textsf{$R$-CIST} of size $k$.
\end{theorem}

\begin{proof}
Let $\mathcal{T} := \{T_1, \dots, T_k\}$ be a \textsf{CISA} in some directed $R$-minor $H$ of $G$. Since every vertex $v \in V(H)$ is a subset of $V(G)$ such that the induced subgraph $G[v]$ is connected, we can construct a spanning tree $Y_v$ in $G[v]$. Let $\phi: V(G) \mapsto V(H)$ be the function that defines the partition of $V(G)$ for which $H$ is the directed $R$-minor (\Cref{def: directed minor}). For each $T_i$, define, in graph $G$, the tree $T_i' := \bigcup_{v \in \mathrm{int}(T_i)} Y_v \cup \{uv: (\phi(u), \phi(v)) \in A(T_i)\}$. Then, prune $T_i'$ to ensure that all leaves are terminals. By this, note that $T_i'$ is an $R$-Steiner tree where $u \in \mathrm{int}(T_i')$ only if $\phi(u) \in \mathrm{int}(T_i)$ (since $L(T_i) \subseteq L(T_i')$), and $E(T_i') \subseteq E(T_i) \cup \bigcup_{v \in \mathrm{int}(T_i)}E(G[v])$. Since $\mathcal{T}$ is a \textsf{CISA}, we have $\mathrm{int}(T_i) \cap \mathrm{int}(T_j) = \emptyset$ and $E(T_i) \cap E(T_j) = \emptyset$ for all distinct $i,j \in [k]$. This implies that $\mathrm{int}(T_i') \cap \mathrm{int}(T_j') = \emptyset$ and $E(T_i') \cap E(T_j') = \emptyset$ for all distinct $i,j \in [k]$. Therefore, $\{T_1', \dots, T_k'\}$ is an \textsf{$R$-CIST} of $G$.
\end{proof}

From \Cref{thm: CISA rcist 2,theorem: CISA bwd}, we get the following corollary.

\begin{corollary}
    A graph $G$ has an \textsf{$R$-CIST} of size 2 if and only if there exists a directed $R$-minor of $G$ that has a \textsf{CISA} of size 2.
\end{corollary}

\subsubsection{\textsf{CISA} in General Digraphs}

Motivated by \Cref{theorem: CISA bwd,theorem: CISA fwd}, we take a detour to explore more about \textsf{CISA} in general digraphs. While the problems of spanning-tree packing and node/edge-independent spanning trees in undirected graphs have long since been extended to their directed variants \cite{edmonds1973edge,Huck1994}, the problem of \textsf{CIST} has not yet been extended to digraphs. There has also been work extending the concepts of internally disjoint and edge-disjoint Steiner trees to digraphs, where each Steiner tree is required to be an arborescence rooted at a terminal vertex \cite{Cheriyan_directed,directed_pendant}. 
To the best of our knowledge, our paper is the first to study completely independent spanning arborescences in digraphs. 

In \Cref{thm: CISA partition}, we give a theorem analogous to the \textsf{CIST}-partition theorem (\Cref{theorem: araki-cist-partition}), for which we will need the following definitions. We call a digraph $G$ \emph{root-connected} if there exists a vertex $v$, called the \emph{root}, such that for every vertex $w$ of $G$, there exists a directed $(v,w)$-path in $G$. If $V_1$ and $V_2$ are disjoint subsets of $V(G)$, we write $B(V_1,V_2,G)$ to denote the sub-digraph of $G$ consisting of all arcs with one end in $V_1$ and the other end in $V_2$. 

\begin{theorem}\label{thm: CISA partition}
    A digraph $G$ has a \textsf{CISA} of size $k$ if and only if there exists a partition $\{V_1, \dots, V_k\}$ of $V(G)$ such that:
    \begin{enumerate}
        \item $G[V_i]$ is root-connected for all $i \in [k]$; and
        \item $\delta^{-}(B(V_i,V_j,G)) \geq 1$ for all distinct $i,j \in [k]$
    \end{enumerate}
\end{theorem}

\begin{proof}
We first prove the forward direction. Let $\{T_1, \dots, T_k\}$ be a \textsf{CISA} of $G$. Define $V_i := \mathrm{int}(T_i)$ for all $i \in [k]$. If there exists a vertex $v \notin \bigcup_{i \in [k]} V_i$, then $v$ must be a leaf of all the trees, and we assign $v$ arbitrarily to some set $V_i$. With this, we get a partition $\{V_1, \dots, V_k\}$. We now prove that this partition satisfies the two conditions. The first condition clearly holds since $G[V_i]$ has an arborescence that is a connected subtree of $T_i$. Now, consider any digraph $B(V_i,V_j,G)$. Let $v \in V_i$ be a vertex. Since $T_j$ is an arobrescence, and $v$ is a leaf of $T_j$, there is a vertex $w \in T_j$ such that $(w,v) \in A(G)$. Likewise, every vertex of $V_j$ has an in-neighbour from some vertex of $V_i$. Therefore, $\delta^{-}(B(V_i,V_j,G)) = 1$.

Now, we show the reverse direction. For each $V_i$, there exists an arborescence $T_i'$ in $G[V_i]$, since $G[V_i]$ is root-connected. For any vertex $v \in V_j \setminus V_i$, the minimum in-degree condition implies that there exists a vertex $w \in V_i$ such that $(w,v) \in E(G)$. Therefore, we can add $v$ as a leaf to $T_i'$. Having done this for all vertices $v \in V(G) \setminus V_i$, we get a spanning arborescence $T_i$ of $G$. Since $\mathrm{int}(T_i) \cap \mathrm{int}(T_j) = \emptyset$ and $E(T_i) \cap E(T_j) = \emptyset$ for all distinct $i,j \in [k]$, $\{T_1, \dots, T_k\}$ is a \textsf{CISA} in $G$.
\end{proof}

In the literature on \textsf{CIST}, it has been observed that many graph properties sufficient to guarantee Hamiltonicity are also sufficient to ensure the existence of \textsf{CIST}; see \cite{araki_dirac,fan_condition,Ore,deg_sum} for examples. It is therefore natural to ask whether a similar phenomenon holds for \textsf{CISA}. In \Cref{theorem: semidegree}, we show that this is at least true for the Ghouila-Houri theorem \cite{ghouilahouri}. This is an analogue of Dirac's theorem for undirected graphs, and is stated as follows: If $G$ is a digraph such that $\delta^0(G) \geq |V(G)|/2$, then $G$ has a directed Hamiltonian cycle (a cycle $C$ such that $V(C) = V(G)$ and $d_C^-(v) = d_C^+(v) = 1$).

\begin{theorem}\label{theorem: semidegree}
    If $G$ is a digraph such that $\delta^0(G) \geq |V(G)|/2$, then $G$ has a \textsf{CISA} of size 2.
\end{theorem}

\begin{proof}
By Ghouila-Houri theorem, we know that $G$ has a directed Hamiltonian cycle $C := (v_1,\dots,v_n)$, where $n := |V(G)|$. Let $\{V_1,V_2\}$ be a partition of $V(G)$ such that $V_1 := \{v_1, \dots, v_{\floor{n/2}}\}$ and $V_2 = \{v_{\floor{n/2}+1}, \dots, v_{n}\}$. Clearly, $G[V_1]$ and $G[V_2]$ are root-connected due to $v_1$ and $v_{\floor{n/2}+1}$, respectively. By the semi-degree condition $\delta^0(G) \geq n/2$, we have $\delta^-(B(V_1,V_2,G)) \geq 1$. Therefore, by \Cref{thm: CISA partition}, $G$ has a \textsf{CISA} of size 2. 
\end{proof}

\section{Vertex Connectivity and Completely Independent Steiner Trees}\label{sec: connectivity}

In this section, we explore the connection between vertex connectivity of a graph $G$ and the size of a maximum \textsf{$R$-CIST} of $G$. The \textit{vertex connectivity} of a graph $G$, denoted by $\kappa(G)$, is the size of a smallest set of vertices $S$ such that $G-S$ is either a disconnected graph, or is an isolated vertex. 
In \cite{Hager}, Hager showed that for any graph $G$ with $\kappa(G) \geq 2^{3p(r-1)}$, for any set of $r$ terminals in $G$, there exist a \textsf{Pendant $R$-CIST} of size $p$ in $G$. The main idea behind this result lies in a reduction to the $k$-linkage problem. A graph $G$ is \emph{$k$-linked} if it has at least $2k$ vertices, and for any sequence of $2k$ distinct vertices $(x_1, \dots, x_k,y_1,\dots,y_k)$, there exist $k$ vertex-disjoint paths with end vertices $x_i$ and $y_i$. The exponential bound in Hager’s result is not intrinsic to the technique itself, but reflects the state of knowledge at the time, when the dependence of linkage on connectivity was only known to be exponential. However, with improved bounds on connectivity for $k$-linkage, we can better the result obtained by Hager. In fact, we prove a more general result in \Cref{theorem: connectivity and rcsit}, that incorporates both pendant and non-pendant $R$-Steiner trees. Before this, we will need \Cref{def: linkage,lemma: linkage}.

\begin{definition}[$H$-linkage]\label{def: linkage}
 A graph $G$ is \emph{$H$-linked} if, for any injective function $f: V(H) \mapsto V(G)$, and any edge $e = uv \in E(H)$, there exists an $(f(u),f(v))$-path in $G$, say $P_e$, such that for any two distinct edges $e_1,e_2 \in E(H)$, the paths $P_{e_1}$ and $P_{e_2}$ are internally vertex-disjoint.
\end{definition}  

\begin{lemma}\label{lemma: linkage}
If $G$ and $H$ are two connected graphs such that $\kappa(G) \geq 10|E(H)|$, then $G$ is $H$-linked.
\end{lemma}

\begin{proof}
Consider any injective function $f: V(H) \mapsto V(G)$, and for each $v_i \in V(H)$, let $w_i := f(v_i)$. As $\delta(G) \geq \kappa(G) \geq 10|E(H)| \geq |V(H)| + 2|E(H)|$, there exist disjoint sets of vertices $\{W_1, \dots, W_{|V(H)|}\}$ such that $\{w_i\} \subseteq W_i \subseteq N_G[w_i]$ and $|W_i| = d_H(v_i)+1$. (Note that $\sum_i|W_i| = |V(H)| + 2|E(H)|$.) Since $\kappa(G) \geq 10|E(H)|$, the graph $G$ is $|E(H)|$-linked \cite{thomas_wollan}. Now, we construct a sequence $\{x_1, \dots, x_{|E(H)|}, \allowbreak y_1, \dots, y_{|E(H)|}\}$ of $2|E(H)|$ vertices of $G$ as follows: Consider any ordering $\{e_1, \dots, e_{|E(H)|}\}$ of the edges in $E(H)$, and for each edge $e_i := v_jv_k \in E(H)$, set $x_i$ to be a vertex from $W_j\setminus \{w_j,x_1, \dots, x_{i-1}, \allowbreak y_1, \dots, y_{i-1}\}$ and $y_i$ to be a vertex from $W_k \setminus \{w_k,x_1, \dots, x_{i-1},\allowbreak y_1, \dots, y_{i-1}\}$. (The sizes of the sets $\{W_1, \dots, W_{|V(H)|}\}$ make this feasible.) Since $G$ is $|E(H)|$-linked, there exist vertex-disjoint $(x_i,y_i)$-paths in $G$, for all $i \in [|E(H)|]$. Finally, we connect each $w_i$ to all vertices in $W_i \setminus \{w_i\}$. This gives us the required subdivision of $H$ in $G$.
\end{proof}

By \Cref{lemma: linkage}, if one can construct a graph $H$ with few edges so that $H$ has an \textsf{$R$-CIST} of size $k$, then any graph $G$ with $\kappa(G) \geq 10|E(H)|$ will also have an \textsf{$R$-CIST} of size $k$. We use this idea, and prove a slightly more general result in \Cref{theorem: connectivity and rcsit}.

\begin{theorem}\label{theorem: connectivity and rcsit}
Let $p,q,r$ be three positive integers where $q \leq r$. If $G$ is a graph with $\kappa(G) \geq 10\{(p+q)r - q\}$, then, for any choice of $r$ terminals, $G$ has an \textsf{$R$-CIST} consisting of $p$ pendant and $q$ non-pendant $R$-Steiner trees.
\end{theorem}

\begin{proof}
We will build a graph $H$ with a set of terminals $R_H$, where $|R_H| = r$, and show that $H$ has an \textsf{$R_H$-CIST} consisting of $p$ pendant and $q$ non-pendant $R_H$-Steiner trees. Then, we will show that $G$ is $H$-linked, and this will establish the theorem. 

To build $H$, we start with two disjoint vertex sets $\{P_H, R_H\}$ such that $|P_H| = p$ and $|R_H| = r$. Then, from each vertex in $P_H$, we add an edge to every vertex in $R_H$. Next, we select a subset $Q_H \subseteq R_H$ such that $|Q_H| = q$ (feasible because $q \leq r$), and add an edge from each vertex of $Q_H$ to every other vertex of $R_H$. Let this multi-graph be $H$. Note that $|E(H)| = pr + q(r - 1) = (p+q)r - q$. Since $\kappa(G) \geq 10\{(p+q)r - q\}$, the graph $G$ is $H$-linked (\Cref{lemma: linkage}). Note that $H$ has an \textsf{$R_H$-CIST} consisting of $p$ pendant $R_H$-Steiner trees, each of which is a star with a distinct internal node in $P_H$, and $q$ non-pendant $R_H$-Steiner trees, each of which is a star with a distinct internal node in $Q_H$.
\end{proof}

\section{Algorithms and Hardness Results}
\label{sec: algorithms}

In this section, we focus on algorithms and hardness results for computing \textsf{$R$-CIST} in graphs. Let \textsf{R-CIST(G,r,k)} be the following decision problem: ``Given a graph $G$ with a set $R$ of $r$ terminals, does $G$ have an \textsf{$R$-CIST} of size $k$?''. In \Cref{thm: hardness}, we show that \textsf{R-CIST(G,r,k)} is NP-complete by reducing from \textsf{R-DST(G,r,k)}, the corresponding decision problem defined for \textsf{$R$-DST}.

\begin{theorem}\label{thm: hardness}
For any $r \geq 3$, \textsf{R-CIST(G,r,k)} is NP-complete. For any $k \geq 2$, \textsf{R-CIST(G,r,k)} is NP-complete.  
\end{theorem}

\begin{proof}
We first show that the problem is in NP. Given a set of trees as a certificate, we can use \Cref{theorem: disjoint internal nodes} to verify if it is an $R$-CIST: We traverse each tree, making sure that it is an $R$-Steiner tree, and mark all edges and internal nodes. If no edge or internal node appears in more than one tree, then the set of trees is an $R$-CIST. 

To show NP-hardness, we reduce from the problem \textsf{R-DST(G,r,k)}: ``Given a graph $G$ with a set $R$ of $r$ terminals, does $G$ have an \textsf{$R$-DST} of size $k$?'' It is known that for any $r \geq 3$, \textsf{R-DST(G,r,k)} is NP-complete, and for any $k \geq 2$, \textsf{R-DST(G,r,k)} is NP-complete \cite{cheriyan,linkage_general,k_3_npc}. We use a simple reduction from $\textsf{R-DST(G,r,k)}$ to $\textsf{R-CIST(G',r,k)}$: Subdivide every edge of $G$ incident to a terminal vertex exactly once. Let this intermediate graph be $H$. (For the remainder of the proof, we will refer to vertices of $V(H) \setminus V(G)$ as `subdivision vertices'.) Next, for every terminal $v \in R$, add edges so that $N_H[v]$ becomes a clique. Let this graph by $G'$. We now show that there exists a set of trees $\{T_1, \dots, T_k\}$ that is an \textsf{$R$-DST} of $G$ if and only if there exists a set of trees $\{T_1', \dots, T_k'\}$ that is an \textsf{$R$-CIST} of $G'$. 

For the forward direction, we obtain $T_i'$ from $T_i$ as follows: Initially, subdivide each edge of $T_i$ incident to a terminal vertex so that we have a tree that is a subgraph of $G'$. Now, repeat the following iteratively for each terminal $v$: delete $v$, and connect the neighbours of $v$ in $T_i$ by a path using edges of the clique at $N_{G'}(v)$ (These edges are distinct for each tree $T_i$ since no subdivision vertex can belong to two distinct trees $T_i$ and $T_j$.). Finally, add back $v$ as a leaf, and this gives us the tree $T_i$'. Since all terminals are leaves in $\{T_1', \dots, T_k'\}$, we get a (pendant) \textsf{$R$-CIST} of size $k$. 

For the reverse direction, we show how to obtain $T_i$ from $T_i'$: For every edge $xy$ in $T_i'$, where $\{x,y\} \subseteq N_{G'}(v)$, for some terminal $v$, replace this edge by the path $(x,v,y)$. With this, we get a graph $U_i$ (possibly having cycles) that spans all terminals. Nonetheless, the graphs $U_i$ remain edge-disjoint, since no subdivision vertex can belong to two distinct trees $T_i'$ and $T_j'$. Next, we smooth all edges of $U_i$ to remove subdivision vertices. That is, for every path $(a,s,b)$, where $s \notin V(G)$, the vertex $s$ has degree 2 in $G'$; therefore, we can replace the path $(a,s,b)$ by the edge $ab$ in $G$. Finally, take any spanning tree of this graph, and prune it to obtain an $R$-Steiner tree $T_i$. Since the edges and non-terminal nodes of $\{T_1, \dots, T_k\}$ are disjoint, $\{T_1, \dots, T_k\}$ is an \textsf{$R$-DST} of $G$.
\end{proof}

Since the proof of \Cref{thm: hardness} reduces an \textsf{$R$-DST} of size $k$ to a \textsf{Pendant $R$-CIST} of size $k$, we get the following corollary: 

\begin{corollary}
For any $r \geq 3$, \textsf{Pendant R-CIST(G,r,k)} is NP-complete. For any $k \geq 2$, \textsf{Pendant R-CIST(G,r,k)} is NP-complete.
\end{corollary}

In \cite{cheriyan_talg}, it was shown that the optimisation version of \textsf{R-DST(G,r,k)} is NP-hard to approximate within a factor of $\Omega(\log n)$, and in \cite{cheriyan}, it was shown to be APX-hard even for $r = 3$. Using the reductions in the proof of \Cref{thm: hardness}, we infer that the same hardness results are true for the optimisation versions of both \textsf{R-CIST(G,r,k)} and \textsf{Pendant R-CIST(G,r,k)}. 

\medskip
In \Cref{thm: polytime enumeration}, we give a polynomial time algorithm for \textsf{R-CIST(G,r,k)} when both $r$ and $k$ are constants. The main idea behind this is that we can use a polynomial time algorithm for deciding whether, given a constant-length sequence of distinct vertices $(x_1, \dots, x_q,y_1,\dots,y_q)$, there exist vertex-disjoint $(x_i,y_i)$-paths for all $i \in [q]$. The algorithm presented here is inspired from similar algorithms in the literature on \textsf{$R$-DST} \cite{linkage_general,linkage_planar}. 
It is worth mentioning that for $r=2$, one may use a simple max-flow algorithm for \textsf{R-CIST(G,r,k)}, since any \textsf{$R$-CIST} is a collection of vertex-disjoint paths between the two terminals.

\begin{figure}
\centering

\begin{subfigure}[t]{0.23\textwidth}
    \centering
    \includegraphics[scale = 0.6, page = 1]{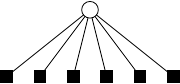}
\end{subfigure}
\hfill
\begin{subfigure}[t]{0.23\textwidth}
    \centering
    \includegraphics[scale = 0.6, page = 2]{Figures/Templates.pdf}
\end{subfigure}
\hfill
\begin{subfigure}[t]{0.23\textwidth}
    \centering
    \includegraphics[scale = 0.6, page = 5]{Figures/Templates.pdf}
\end{subfigure}
\hfill
\begin{subfigure}[t]{0.23\textwidth}
    \centering
    \includegraphics[scale = 0.6, page = 6]{Figures/Templates.pdf}
\end{subfigure}

\vspace{1em}

\hspace{0.23\textwidth}
\hfill
\begin{subfigure}[t]{0.23\textwidth}
    \centering
    \includegraphics[scale = 0.6, page = 3]{Figures/Templates.pdf}
\end{subfigure}
\hfill
\begin{subfigure}[t]{0.23\textwidth}
    \centering
    \includegraphics[scale = 0.6, page = 4]{Figures/Templates.pdf}
\end{subfigure}
\hfill
\begin{subfigure}[t]{0.23\textwidth}
    \centering
    \includegraphics[scale = 0.6, page = 7]{Figures/Templates.pdf}
\end{subfigure}

\caption{This figure shows the set of all non-isomorphic templates on 6 terminals that generate pendant $R$-Steiner trees.}
\label{fig: templates}
\end{figure}

\begin{theorem}\label{thm: polytime enumeration}
   If both $r$ and $k$ are fixed integers, then \textsf{R-CIST(G,r,k)} is decidable in polynomial time. 
\end{theorem}

\begin{proof}
The algorithm proceeds with the following steps:

\begin{enumerate}
\setlength{\itemsep}{10pt}

    \item \textbf{Generating Templates:} The first step involves generating trees whose vertices are of two kinds: \emph{black} and \emph{white}. (The black vertices are intended for terminals, and the white vertices for non-terminals.) All leaves of the tree must be black, and every white vertex must have degree at least 3. The total number of black vertices should exactly equal $r$. Furthermore, no two trees $T_i$ and $T_j$ must be isomorphic in the following sense: There is a bijection $f: V(T_1) \mapsto V(T_2)$ that maps black vertices of $T_1$ to black vertices of $T_2$, and $uv \in E(T_1)$ if and only if $f(u)f(v) \in E(T_2)$. We call each of these trees as a \emph{template}, and let them be denoted by $\{\tau_1, \dots, \tau_p\}$. (See \Cref{fig: templates} for an example. Note that the figure only shows templates for pendant $R$-Steiner trees.) Since any template has exactly $r$ vertices of degree at most 2, it has at most $2r$ vertices. This implies that the total number of non-isomorphic templates $p$ is some function of $r$, and is therefore a constant. Hence, one can generate $\{\tau_1, \dots, \tau_p\}$ in $O(1)$ time.
    
    \item \textbf{Generating Candidate Solutions:} Next, consider all possible distinct sequences of non-negative integers  $(k_1, \dots, k_p)$ such that $k_1 + \dots + k_p = k$. For each $i \in [p]$, we would like to find $k_i$ trees of type $\tau_i$ so that the entire set of $k$ trees is an \textsf{$R$-CIST}. For this, we consider all possible ways of labelling the $k$ trees with vertices of $G$ subject to the following rules: Every black vertex of a tree must be labelled by a distinct terminal, every white vertex by a distinct non-terminal, and any node that is internal to one tree should not be internal to any other tree. We call each such set of $k$ labelled trees as a \emph{candidate solution}. Since the total number of white vertices in all the $k$ trees is a constant, and no two white vertices can get labelled with the same vertex of $G$, the total number of candidate solutions is $n^{O(1)}$, where $n = |V(G)|$.
    
    \item \textbf{Running Disjoint-Paths Algorithm:} If we succeed in finding a candidate solution, then we construct an auxiliary graph $G'$ as follows: For each vertex $v \in V(G)$ in the candidate solution, add $\sum_{\tau_i} k_i \cdot d_{\tau_i}(v)$ \emph{twins} of $v$---these are vertices that have the same neighbourhood as $v$ in $G$. By this, we can construct a sequence $(x_1, \dots, x_q,y_1,\dots,y_q)$, where $q := \sum_{\tau_i}k_i|E(\tau_i)|$, such that for any edge $e_i \in \bigcup_{\tau_i}E(\tau_i)$ whose end vertices are labelled $a$ and $b$, there exist distinct vertices $a' =: x_i$ and $b' =: y_i$, where $a'$ and $b'$ are twins of $a$ and $b$ respectively. Then, we run a disjoint-paths algorithm on $G'$ to test whether there exist vertex-disjoint $(x_i,y_i)$-paths in $G'$. We repeat this process for all candidate solutions until we succeed. The time taken to run the disjoint-paths algorithm is $O(|V(G')|^2)$, since $q \in O(1)$ \cite{dpp_quadratic}. Since $k,p \in O(1)$, we have $|V(G')| \in O(n)$, and therefore, the disjoint-paths algorithm runs in $O(n^2)$ time. As the number of candidate solutions is $n^{O(1)}$, the total run-time for this algorithm is at most $n^{O(1)}$.\qedhere
\end{enumerate}
\end{proof}

In \Cref{cor: planar runtime,cor: non-pendant runtime}, we give some simple applications of \Cref{thm: polytime enumeration}.

\begin{corollary}\label{cor: planar runtime}
    If $G$ is a planar graph with $|R|$ terminals, where $3 \leq |R|$ and $|R| \in O(1)$, then one can compute a maximum \textsf{$R$-CIST} of $G$ in polynomial time. 
\end{corollary}

\begin{proof}
From \Cref{theorem: planar graphs}, we know that any \textsf{$R$-CIST} of $G$ has size at most 5. Therefore, we can apply \Cref{thm: polytime enumeration} for all values of $k$ in the range $\{2,\dots,5\}$ to obtain a maximum \textsf{$R$-CIST} of $G$.   
\end{proof}

\begin{corollary}\label{cor: non-pendant runtime}
If $G$ is a graph with $|R|$ terminals, where $3 \leq |R|$ and $|R| \in O(1)$, then one can compute a maximum \textsf{Non-pendant} \textsf{$R$-CIST} of $G$ in polynomial time.
\end{corollary}

\begin{proof}
From \Cref{obs: number of RCIST}, a maximum \textsf{Non-pendant} \textsf{$R$-CIST} of $G$ has size at most $|R|$, and is therefore a constant. Then, we can run the algorithm in \Cref{thm: polytime enumeration}, but with a small change: we ensure that each template contains at least one interior black vertex. This will then guarantee that all solutions consist of only non-pendant $R$-Steiner trees.  
\end{proof}

For graphs of bounded treewidth, we can get a linear time algorithm to decide \textsf{R-CIST(G,r,k)} for constant values of $k$, since we can frame this question in the language of Monadic Second Order Logic (MSOL), and use Courcelle's theorem.  

\begin{theorem}[Courcelle's theorem \cite{courcelle1990monadic}]\label{thm: courcelle}
    Every graph property definable in the monadic second-order logic of graphs can be decided in linear time on graphs of bounded treewidth.
\end{theorem}

MSOL is an extension of first-order logic in which one can quantify over vertices and edges (e.g. $\exists v \in V$, $\forall e \in E$), and over vertex-sets and edge-sets ($\exists X \subseteq V$, $\exists F \subseteq E$). The allowed operations include logical connectives ($\neg$, $\land$, $\lor$, $\to$), equality, and basic relations such as adjacency between vertices. However, MSOL does not allow quantification over arbitrary relations of higher arity (only unary set variables are permitted), nor does it allow arithmetic operations like addition or multiplication on indices, or counting arbitrary cardinalities (unless using special counting extensions like CMSOL). In \Cref{thm: msol}, we show that if $k$ is a constant, then \textsf{R-CIST(G,r,k)} can be expressed in MSOL, and hence can be decided in linear time for graphs of bounded treewidth.

\begin{theorem}\label{thm: msol}
   If $k$ is a fixed integer and $G$ is a graph of bounded treewidth, then \textsf{R-CIST(G,r,k)} is decidable in linear time.
\end{theorem}

\begin{proof}
Given an edge set $E_i \subseteq E(G)$, we assume that we have a formula $\text{Tree}(E_i)$ that is true if and only if $G[E_i]$ is a tree, $\text{Inc}(v,E_i)$ that is true if and only if $v$ is a vertex of $G[E_i]$, and $\text{Deg}_1(v,E_i)$ that is true if and only if $v$ has degree 1 in $G[E_i]$. The formulae for these can be found in \cite{bodlaender}.

\medskip
The formula below encodes that a set of edges $E_i$ induces an $R$-Steiner tree:
$$ \text{$R$-ST}(E_i) := \text{Tree}(E_i) \land (\forall v \in R)(\text{Inc}(v,E_i)) \land (\forall v \in V)(\text{Deg}_1(v,E_i) \to  v \in R)$$
The next formula encodes that a set of edges $\{E_1, \dots, E_k\}$ induce edge-disjoint $R$-Steiner trees:
$$ \text{$R$-EST}(E_{[k]}) := \bigwedge_{i \in [k]}\text{$R$-ST}(E_i) \land (\forall e \in E)\Big(\bigwedge_{i \neq j}(e \notin E_i) \lor (e \notin E_j)\Big) $$
The next formula encodes that a set of edges $\{E_1, \dots, E_k\}$ induces an \textsf{$R$-CIST}:
$$ \text{$R$-CIST}(E_{[k]}) := \text{$R$-EST}(E_{[k]}) \land \bigwedge_{i \neq j}(\mathrm{Inc}(v,E_i)\land\mathrm{Inc}(v,E_j))
\to\mathrm{Deg}_1(v,E_j)\lor\mathrm{Deg}_1(v,E_i)$$
Finally, $G$ has an \textsf{$R$-CIST} of size $k$ if and only if the following formula is true: $$\text{$R$-CIST}(G,k) := \exists(E_1, \dots, E_k)(\text{$R$-CIST}(E_{[k]}))$$

Since $k$ is a constant, all the formulae above have constant length. Therefore, by \Cref{thm: courcelle}, \textsf{R-CIST(G,r,k)} is decidable in linear time for constant $k$.
\end{proof}

In \Cref{cor: bdd tw np rcist runtime,cor: bdd tw rcist runtime}, we give some simple applications of \Cref{thm: msol}.

\begin{corollary}\label{cor: bdd tw np rcist runtime}
    If $G$ is a graph of bounded treewidth, then one can compute the maximum size of a \textsf{Non-pendant $R$-CIST} of $G$ in linear time. 
\end{corollary}

\begin{proof}
We know from \Cref{theorem: bounded treewidth} that any \textsf{Non-pendant $R$-CIST} of $G$ has size at most $tw(G) + 1$. To decide whether a \textsf{Non-pendant $R$-CIST} has size $k$, we can use the same formulas as listed above in the proof of \Cref{thm: msol}, but with a small change. To ensure that we only search for non-pendant $R$-Steiner trees, instead of $\text{$R$-ST}(E_i)$, we use the formula 
$$\text{$npR$-ST}(E_i) := \text{$R$-ST}(E_i) \land (\exists v \in R)(\neg \text{Deg}_1(v,E_i))$$
To find the size of a maximum \textsf{Non-pendant $R$-CIST}, we simply search for the largest value of $k$ in the range $\{2, \dots, tw(G)+1\}$ for which the formula $\text{$R$-CIST}(G,k)$ is true.
\end{proof}

\begin{corollary}\label{cor: bdd tw rcist runtime}
    If a graph $G$ has treewidth $w \in O(1)$, then for any set of terminals $R$ where $|R| \geq w+1$, one can compute the maximum size of an \textsf{$R$-CIST} of $G$ in linear time.
\end{corollary}

\begin{proof}
    From \Cref{cor: bounded tw bounded rcist}, we know that any \textsf{$R$-CIST} of $G$ has size at most $2w$. Therefore, we simply search for the largest value of $k$ in the range $\{2, \dots, tw(G)+1\}$ for which the formula $\text{$R$-CIST}(G,k)$ is true.
\end{proof}

\section{Avenues for Future Work}\label{sec: conclusion}

This paper surveys the landscape of completely independent Steiner trees in general graphs, covering various characterizations, algorithms, and hardness results. For future work, it would be interesting to investigate completely independent Steiner trees for special graph classes, classical interconnection networks (hypercubes and their variants, torus and mesh networks, line graphs, etc.), and data center networks (DCell, BCube, etc.).  Another natural direction is to study the existence of completely independent Steiner trees of small diameter, which is particularly relevant in interconnection networks for communication efficiency; see \cite{locally_twisted} for related work on completely independent spanning trees.

This paper also introduced the notion of completely independent spanning arborescences. Consequently, one could further explore structural properties and characterisations of such objects in directed graphs---a useful starting point would be the surveys on packing spanning arborescences in digraphs \cite{berczi2009packing,szigeti2017survey}.

\newpage
\bibliography{bib2doi}

@article{Hager,
  author = {Michael Hager},
  title = {Pendant tree-connectivity},
  journal = {J. Comb. Theory {B}},
  volume = {38},
  number = {2},
  pages = {179--189},
  year = {1985},
  url = {https://doi.org/10.1016/0095-8956(85)90083-8},
  doi = {10.1016/0095-8956(85)90083-8},
  timestamp = {Fri, 07 Jun 2024 01:00:00 +0200},
  biburl = {https://dblp.org/rec/journals/jct/Hager85.bib},
  bibsource = {dblp computer science bibliography, https://dblp.org},
  _bib2doi_selected = {dblp:/rec/journals/jct/Hager85.bib},
  _bib2doi_confirmed = {true},
}

@article{dpp_quadratic,
  title = {The disjoint paths problem in quadratic time},
  journal = {J. Comb. Theory {B}},
  volume = {102},
  number = {2},
  pages = {424--435},
  year = {2012},
  issn = {0095-8956},
  doi = {10.1016/j.jctb.2011.07.004},
  url = {https://doi.org/10.1016/j.jctb.2011.07.004},
  author = {Ken{-}ichi Kawarabayashi and Yusuke Kobayashi and Bruce A. Reed},
  keywords = {Disjoint paths, Quadratic time, Graph minor, Tree-width},
  abstract = {We consider the following well-known problem, which is called the disjoint paths problem. For a given graph G and a set of k pairs of terminals in G, the objective is to find k vertex-disjoint paths connecting given pairs of terminals or to conclude that such paths do not exist. We present an O(n2) time algorithm for this problem for fixed k. This improves the time complexity of the seminal result by Robertson and Seymour, who gave an O(n3) time algorithm for the disjoint paths problem for fixed k. Note that Perković and Reed (2000) announced in [24] (without proofs) that this problem can be solved in O(n2) time. Our algorithm implies that there is an O(n2) time algorithm for the k edge-disjoint paths problem, the minor containment problem, and the labeled minor containment problem. In fact, the time complexity of all the algorithms with the most expensive part depending on Robertson and Seymourʼs algorithm can be improved to O(n2), for example, the membership testing for minor-closed class of graphs.},
  timestamp = {Fri, 07 Jun 2024 01:00:00 +0200},
  biburl = {https://dblp.org/rec/journals/jct/KawarabayashiKR12.bib},
  bibsource = {dblp computer science bibliography, https://dblp.org},
  _bib2doi_old_doi = {https://doi.org/10.1016/j.jctb.2011.07.004},
  _bib2doi_selected = {dblp:/rec/journals/jct/KawarabayashiKR12.bib},
  _bib2doi_confirmed = {true},
  _bib2doi_finished = {true},
}

@article{cheriyan,
  author = {Ashkan Aazami and Joseph Cheriyan and Krishnam Raju Jampani},
  title = {Approximation Algorithms and Hardness Results for Packing Element-Disjoint {S}teiner Trees in Planar Graphs},
  journal = {Algorithmica},
  volume = {63},
  number = {1-2},
  pages = {425--456},
  year = {2012},
  url = {https://doi.org/10.1007/s00453-011-9540-3},
  doi = {10.1007/s00453-011-9540-3},
  timestamp = {Wed, 17 May 2017 01:00:00 +0200},
  biburl = {https://dblp.org/rec/journals/algorithmica/AazamiCJ12.bib},
  bibsource = {dblp computer science bibliography, https://dblp.org},
  _bib2doi_selected = {dblp:/rec/journals/algorithmica/AazamiCJ12.bib},
  _bib2doi_confirmed = {true},
}

@article{hwang1995steiner,
  title = {The {S}teiner tree problem, annals of discrete mathematics, volume 53},
  author = {Hwang, Frank K and Richards, Dana S and Winter, Pawel and Widmayer, P},
  journal = {ZOR-Methods and Models of Operations Research},
  volume = {41},
  number = {3},
  pages = {382},
  year = {1995},
  publisher = {Wurzburg, Physica-Verlag, 1972-1995.},
  _bib2doi_finished = {true},
}

@article{LU200355,
  title = {The full {S}teiner tree problem},
  journal = {Theor. Comput. Sci.},
  volume = {306},
  number = {1-3},
  pages = {55--67},
  year = {2003},
  issn = {0304-3975},
  doi = {10.1016/S0304-3975(03)00209-3},
  url = {https://doi.org/10.1016/S0304-3975(03)00209-3},
  author = {Chin Lung Lu and Chuan Yi Tang and Richard Chia{-}Tung Lee},
  keywords = {Full Steiner tree problem, Phylogenetic tree, Evolutionary tree, NP-complete, MAX SNP-hard, Approximation algorithm},
  abstract = {Motivated by the reconstruction of phylogenetic tree in biology, we study the full Steiner tree problem in this paper. Given a complete graph G=(V,E) with a length function on E and a proper subset R⊂V, the problem is to find a full Steiner tree of minimum length in G, which is a kind of Steiner tree with all the vertices of R as its leaves. In this paper, we show that this problem is NP-complete and MAX SNP-hard, even when the lengths of the edges are restricted to either 1 or 2. For the instances with lengths either 1 or 2, we give a 85-approximation algorithm to find an approximate solution for the problem.},
  timestamp = {Wed, 17 Feb 2021 00:00:00 +0100},
  biburl = {https://dblp.org/rec/journals/tcs/LuTL03.bib},
  bibsource = {dblp computer science bibliography, https://dblp.org},
  _bib2doi_old_doi = {https://doi.org/10.1016/S0304-3975(03)00209-3},
  _bib2doi_selected = {dblp:/rec/journals/tcs/LuTL03.bib},
  _bib2doi_confirmed = {true},
  _bib2doi_finished = {true},
}

@article{hasunuma,
  author = {Toru Hasunuma},
  title = {Completely independent spanning trees in the underlying graph of a line digraph},
  journal = {Discret. Math.},
  volume = {234},
  number = {1-3},
  pages = {149--157},
  year = {2001},
  url = {https://doi.org/10.1016/S0012-365X(00)00377-0},
  doi = {10.1016/S0012-365X(00)00377-0},
  timestamp = {Sun, 19 Jan 2025 00:00:00 +0100},
  biburl = {https://dblp.org/rec/journals/dm/Hasunuma01.bib},
  bibsource = {dblp computer science bibliography, https://dblp.org},
  _bib2doi_selected = {dblp:/rec/journals/dm/Hasunuma01.bib},
  _bib2doi_confirmed = {true},
}

@article{cheriyan_talg,
  author = {Joseph Cheriyan and Mohammad R. Salavatipour},
  title = {Packing element-disjoint {S}teiner trees},
  journal = {{ACM} Trans. Algorithms},
  volume = {3},
  number = {4},
  pages = {47},
  year = {2007},
  url = {https://doi.org/10.1145/1290672.1290684},
  doi = {10.1145/1290672.1290684},
  timestamp = {Tue, 06 Nov 2018 00:00:00 +0100},
  biburl = {https://dblp.org/rec/journals/talg/CheriyanS07.bib},
  bibsource = {dblp computer science bibliography, https://dblp.org},
  _bib2doi_selected = {dblp:/rec/journals/talg/CheriyanS07.bib},
  _bib2doi_confirmed = {true},
}

@book{tree_book,
  title = {Generalized Connectivity of Graphs},
  isbn = {9783319338286},
  issn = {2191-8201},
  url = {http://dx.doi.org/10.1007/978-3-319-33828-6},
  doi = {10.1007/978-3-319-33828-6},
  journal = {SpringerBriefs in Mathematics},
  publisher = {Springer International Publishing},
  author = {Li, Xueliang and Mao, Yaping},
  year = {2016},
  _bib2doi_finished = {true},
}

@article{araki_dirac,
  author = {Toru Araki},
  title = {Dirac's Condition for Completely Independent Spanning Trees},
  journal = {J. Graph Theory},
  volume = {77},
  number = {3},
  pages = {171--179},
  year = {2014},
  url = {https://doi.org/10.1002/jgt.21780},
  doi = {10.1002/jgt.21780},
  timestamp = {Fri, 02 Oct 2020 01:00:00 +0200},
  biburl = {https://dblp.org/rec/journals/jgt/Araki14.bib},
  bibsource = {dblp computer science bibliography, https://dblp.org},
  _bib2doi_selected = {dblp:/rec/journals/jgt/Araki14.bib},
  _bib2doi_confirmed = {true},
}

@article{Ore,
  author = {Genghua Fan and Yanmei Hong and Qinghai Liu},
  title = {Ore's condition for completely independent spanning trees},
  journal = {Discret. Appl. Math.},
  volume = {177},
  pages = {95--100},
  year = {2014},
  url = {https://doi.org/10.1016/j.dam.2014.06.002},
  doi = {10.1016/j.dam.2014.06.002},
  timestamp = {Thu, 11 Feb 2021 00:00:00 +0100},
  biburl = {https://dblp.org/rec/journals/dam/FanHL14.bib},
  bibsource = {dblp computer science bibliography, https://dblp.org},
  _bib2doi_selected = {dblp:/rec/journals/dam/FanHL14.bib},
  _bib2doi_confirmed = {true},
}

@article{deg_sum,
  title = {Degree condition for completely independent spanning trees},
  journal = {Inf. Process. Lett.},
  volume = {116},
  number = {10},
  pages = {644--648},
  year = {2016},
  issn = {0020-0190},
  doi = {10.1016/j.ipl.2016.05.004},
  url = {https://doi.org/10.1016/j.ipl.2016.05.004},
  author = {Xia Hong and Qinghai Liu},
  keywords = {Completely independent spanning tree, Combinatorial problems},
  abstract = {Let T1,T2,…,Tk be spanning trees of a graph G. For any two vertices u,v of G, if the paths from u to v in these k trees are pairwise openly disjoint, then we say that T1,T2,…,Tk are completely independent. Araki showed that a graph G on n≥7 vertices has two completely independent spanning trees if the minimum degree δ(G)≥n/2. In this paper, we give a generalization: a graph G on n≥4k−1 vertices has k completely independent spanning trees if the minimum degree δ(G)≥k−1kn. In fact, we prove a stronger result.},
  timestamp = {Thu, 19 Oct 2017 01:00:00 +0200},
  biburl = {https://dblp.org/rec/journals/ipl/HongL16.bib},
  bibsource = {dblp computer science bibliography, https://dblp.org},
  _bib2doi_old_doi = {https://doi.org/10.1016/j.ipl.2016.05.004},
  _bib2doi_selected = {dblp:/rec/journals/ipl/HongL16.bib},
  _bib2doi_confirmed = {true},
  _bib2doi_finished = {true},
}

@article{fan_condition,
  title = {Fan-type condition for two completely independent spanning trees},
  journal = {Theoretical Computer Science},
  volume = {1064},
  pages = {115718},
  year = {2026},
  issn = {0304-3975},
  doi = {https://doi.org/10.1016/j.tcs.2025.115718},
  url = {https://www.sciencedirect.com/science/article/pii/S0304397525006553},
  author = {Jie Ma and Junqing Cai},
  keywords = {Completely independent spanning trees, CIST-partition, Fan-type condition},
  abstract = {The spanning trees T1,T2,⋯,Tk of a graph G are called k completely independent spanning trees (CISTs) if for any two vertices u, v ∈ V(G), the paths connecting u and v in any two distinct trees are pairwise edge-disjoint and internally vertex-disjoint. CISTs have significant applications in fault-tolerant broadcasting for interconnection networks, substantially improving network reliability and redundancy. However, determining whether a connected graph contains two CISTs is known to be NP-complete. Araki [J. Graph Theory. 77 (2014) 171–179.] posed an open question regarding whether certain known sufficient conditions for hamiltonian cycles could guarantee the existence of two CISTs. In this paper, we provide an affirmative answer to this question by proving that every connected graph G of order n ≥ 7 with μ2(G) ≥ n contains two CISTs. Moreover, both the lower bounds on the order n and the degree condition μ2(G) are best possible.},
  _bib2doi_finished = {true},
}

@article{peterfalvi,
  title = {Two counterexamples on completely independent spanning trees},
  journal = {Discret. Math.},
  volume = {312},
  number = {4},
  pages = {808--810},
  year = {2012},
  issn = {0012-365X},
  doi = {10.1016/j.disc.2011.11.015},
  url = {https://doi.org/10.1016/j.disc.2011.11.015},
  author = {Ferenc P{\'{e}}terfalvi},
  keywords = {Independent spanning trees, Interconnection networks, Plane graphs, Tutte graph},
  abstract = {For each k≥2, we construct a k-connected graph which does not contain two completely independent spanning trees. This disproves a conjecture of Hasunuma. Furthermore, we also give an example of a 3-connected maximal plane graph not containing two completely independent spanning trees.},
  timestamp = {Sun, 19 Jan 2025 00:00:00 +0100},
  biburl = {https://dblp.org/rec/journals/dm/Peterfalvi12.bib},
  bibsource = {dblp computer science bibliography, https://dblp.org},
  _bib2doi_old_doi = {https://doi.org/10.1016/j.disc.2011.11.015},
  _bib2doi_selected = {dblp:/rec/journals/dm/Peterfalvi12.bib},
  _bib2doi_confirmed = {true},
  _bib2doi_finished = {true},
}

@article{thomas_wollan,
  title = {An improved linear edge bound for graph linkages},
  journal = {Eur. J. Comb.},
  volume = {26},
  number = {3-4},
  pages = {309--324},
  year = {2005},
  note = {Topological Graph Theory and Graph Minors, second issue},
  issn = {0195-6698},
  doi = {10.1016/j.ejc.2004.02.013},
  url = {https://doi.org/10.1016/j.ejc.2004.02.013},
  author = {Robin Thomas and Paul Wollan},
  abstract = {A graph is said to be k-linked if it has at least 2k vertices and for every sequence s1,…,sk,t1,…,tk of distinct vertices there exist disjoint paths P1,…,Pk such that the ends of Pi are si and ti. Bollobás and Thomason showed that if a simple graph G on n vertices is 2k-connected and G has at least 11kn edges, then G is k-linked. We give a relatively simple inductive proof of the stronger statement that 8kn edges and 2k-connectivity suffice, and then with more effort improve the edge bound to 5kn.},
  timestamp = {Fri, 12 Feb 2021 00:00:00 +0100},
  biburl = {https://dblp.org/rec/journals/ejc/ThomasW05.bib},
  bibsource = {dblp computer science bibliography, https://dblp.org},
  _bib2doi_old_doi = {https://doi.org/10.1016/j.ejc.2004.02.013},
  _bib2doi_selected = {dblp:/rec/journals/ejc/ThomasW05.bib},
  _bib2doi_confirmed = {true},
  _bib2doi_finished = {true},
}

@article{linkage_general,
  author = {Shasha Li and Xueliang Li},
  title = {Note on the hardness of generalized connectivity},
  journal = {J. Comb. Optim.},
  volume = {24},
  number = {3},
  pages = {389--396},
  year = {2012},
  url = {https://doi.org/10.1007/s10878-011-9399-x},
  doi = {10.1007/s10878-011-9399-x},
  timestamp = {Mon, 13 Nov 2017 00:00:00 +0100},
  biburl = {https://dblp.org/rec/journals/jco/LiL12.bib},
  bibsource = {dblp computer science bibliography, https://dblp.org},
  _bib2doi_selected = {dblp:/rec/journals/jco/LiL12.bib},
  _bib2doi_confirmed = {true},
}

@article{linkage_planar,
  author = {Shasha Li and Xueliang Li and Wenli Zhou},
  title = {Sharp bounds for the generalized connectivity $\kappa_3({G})$},
  journal = {Discret. Math.},
  volume = {310},
  number = {15-16},
  pages = {2147--2163},
  year = {2010},
  url = {https://doi.org/10.1016/j.disc.2010.04.011},
  doi = {10.1016/j.disc.2010.04.011},
  timestamp = {Fri, 12 Feb 2021 00:00:00 +0100},
  biburl = {https://dblp.org/rec/journals/dm/LiLZ10.bib},
  bibsource = {dblp computer science bibliography, https://dblp.org},
  _bib2doi_selected = {dblp:/rec/journals/dm/LiLZ10.bib},
  _bib2doi_confirmed = {true},
  _bib2doi_finished = {true},
}

@article{k_3_npc,
  author = {Lily Chen and Xueliang Li and Mengmeng Liu and Yaping Mao},
  title = {A solution to a conjecture on the generalized connectivity of graphs},
  journal = {J. Comb. Optim.},
  volume = {33},
  number = {1},
  pages = {275--282},
  year = {2017},
  url = {https://doi.org/10.1007/s10878-015-9955-x},
  doi = {10.1007/s10878-015-9955-x},
  timestamp = {Sun, 22 Oct 2023 01:00:00 +0200},
  biburl = {https://dblp.org/rec/journals/jco/ChenLLM17.bib},
  bibsource = {dblp computer science bibliography, https://dblp.org},
  _bib2doi_selected = {dblp:/rec/journals/jco/ChenLLM17.bib},
  _bib2doi_confirmed = {true},
}

@article{bodlaender,
  author = {Hans L. Bodlaender and Fedor V. Fomin and Petr A. Golovach and Yota Otachi and Erik Jan van Leeuwen},
  title = {Parameterized Complexity of the Spanning Tree Congestion Problem},
  journal = {Algorithmica},
  volume = {64},
  number = {1},
  pages = {85--111},
  year = {2012},
  url = {https://doi.org/10.1007/s00453-011-9565-7},
  doi = {10.1007/s00453-011-9565-7},
  timestamp = {Sun, 02 Jun 2019 01:00:00 +0200},
  biburl = {https://dblp.org/rec/journals/algorithmica/BodlaenderFGOL12.bib},
  bibsource = {dblp computer science bibliography, https://dblp.org},
  _bib2doi_selected = {dblp:/rec/journals/algorithmica/BodlaenderFGOL12.bib},
  _bib2doi_confirmed = {true},
}

@article{cisst,
  author = {Jun Yuan and Shan Liu and Shangwei Lin and Aixia Liu},
  title = {Completely independent {S}teiner trees and corresponding tree connectivity},
  journal = {International Journal of Parallel, Emergent and Distributed Systems},
  volume = {0},
  number = {0},
  pages = {1--14},
  year = {2025},
  publisher = {Taylor \& Francis},
  doi = {10.1080/17445760.2025.2609134},
  url = {https://doi.org/10.1080/17445760.2025.2609134},
  eprint = {https://doi.org/10.1080/17445760.2025.2609134},
  _bib2doi_finished = {true},
}

@article{ghouilahouri,
  title = {Une condition suffisante dexistence dun circuit hamiltonien},
  author = {Ghouila-Houri, Alain},
  journal = {Comptes Rendus Hebdomadaires Des Seances De L Academie Des Sciences},
  volume = {251},
  number = {4},
  pages = {495--497},
  year = {1960},
  publisher = {GAUTHIER-VILLARS/EDITIONS ELSEVIER 23 RUE LINOIS, 75015 PARIS, FRANCE},
  _bib2doi_finished = {true},
}

@article{courcelle1990monadic,
  title = {The monadic second-order logic of graphs. {I}. {R}ecognizable sets of finite graphs},
  author = {Courcelle, Bruno},
  journal = {Information and computation},
  volume = {85},
  number = {1},
  pages = {12--75},
  year = {1990},
  publisher = {Elsevier},
  doi = {10.1016/0890-5401(90)90043-H},
  timestamp = {Fri, 12 Feb 2021 00:00:00 +0100},
  biburl = {https://dblp.org/rec/journals/iandc/Courcelle90.bib},
  bibsource = {dblp computer science bibliography, https://dblp.org},
  _bib2doi_selected = {dblp:/rec/journals/iandc/Courcelle90.bib},
  _bib2doi_confirmed = {true},
}

@inproceedings{maximal_planar,
  author = {Toru Hasunuma},
  editor = {Ludek Kucera},
  title = {Completely Independent Spanning Trees in Maximal Planar Graphs},
  booktitle = {Graph-Theoretic Concepts in Computer Science, 28th International Workshop, {WG} 2002, Cesky Krumlov, Czech Republic, June 13-15, 2002, Revised Papers},
  series = {Lecture Notes in Computer Science},
  pages = {235--245},
  publisher = {Springer},
  year = {2002},
  url = {https://doi.org/10.1007/3-540-36379-3\_21},
  doi = {10.1007/3-540-36379-3\_21},
  timestamp = {Fri, 26 May 2017 01:00:00 +0200},
  biburl = {https://dblp.org/rec/conf/wg/Hasunuma02.bib},
  bibsource = {dblp computer science bibliography, https://dblp.org},
  _bib2doi_selected = {dblp:/rec/conf/wg/Hasunuma02.bib},
  _bib2doi_confirmed = {true},
}

@article{survey,
  author = {Baolei Cheng and Dajin Wang and Jianxi Fan},
  title = {Independent Spanning Trees in Networks: {A} Survey},
  journal = {{ACM} Comput. Surv.},
  volume = {55},
  number = {14s},
  pages = {335:1--335:29},
  year = {2023},
  url = {https://doi.org/10.1145/3591110},
  doi = {10.1145/3591110},
  timestamp = {Wed, 24 Jan 2024 00:00:00 +0100},
  biburl = {https://dblp.org/rec/journals/csur/ChengWF23.bib},
  bibsource = {dblp computer science bibliography, https://dblp.org},
  _bib2doi_selected = {dblp:/rec/journals/csur/ChengWF23.bib},
  _bib2doi_confirmed = {true},
}

@article{Tutte1961,
  title = {On the Problem of Decomposing a Graph into n Connected Factors},
  volume = {s1-36},
  issn = {0024-6107},
  url = {http://dx.doi.org/10.1112/jlms/s1-36.1.221},
  doi = {10.1112/jlms/s1-36.1.221},
  number = {1},
  journal = {Journal of the London Mathematical Society},
  publisher = {Wiley},
  author = {Tutte, W. T.},
  year = {1961},
  pages = {221–230},
  _bib2doi_finished = {true},
}

@article{NashWilliams1961,
  title = {Edge-Disjoint Spanning Trees of Finite Graphs},
  volume = {s1-36},
  issn = {0024-6107},
  url = {http://dx.doi.org/10.1112/jlms/s1-36.1.445},
  doi = {10.1112/jlms/s1-36.1.445},
  number = {1},
  journal = {Journal of the London Mathematical Society},
  publisher = {Wiley},
  author = {Nash-Williams, C. St.J. A.},
  year = {1961},
  pages = {445–450},
  _bib2doi_finished = {true},
}

@article{itai,
  title = {The Multi-Tree Approach to Reliability in Distributed Networks},
  journal = {Inf. Comput.},
  volume = {79},
  number = {1},
  pages = {43--59},
  year = {1988},
  issn = {0890-5401},
  doi = {10.1016/0890-5401(88)90016-8},
  url = {https://doi.org/10.1016/0890-5401(88)90016-8},
  author = {Alon Itai and Michael Rodeh},
  abstract = {Consider a network of asynchronous processors communicating by sending messages over unreliable lines. There are many advantages to restricting all communications to a spanning tree. To overcome the possible failure of k′ < k edges, we describe a communication protocol which uses k rooted spanning trees having the property that for every vertex ν the paths from ν to the root are edge-disjoint. An algorithm to find two such trees in a 2-edge connected graph is described that runs in time proportional in the number of edges in the graph. This algorithm has a distributed version which finds the two trees even when a single edge fails during their construction. The two trees then may be used to transform certain centralized algorithms to distributed, reliable, and efficient ones.},
  timestamp = {Fri, 12 Feb 2021 00:00:00 +0100},
  biburl = {https://dblp.org/rec/journals/iandc/ItaiR88.bib},
  bibsource = {dblp computer science bibliography, https://dblp.org},
  _bib2doi_old_doi = {https://doi.org/10.1016/0890-5401(88)90016-8},
  _bib2doi_selected = {dblp:/rec/journals/iandc/ItaiR88.bib},
  _bib2doi_confirmed = {true},
  _bib2doi_finished = {true},
}

@article{chartrand,
  author = {Zehavi, Avram and Itai, Alon},
  title = {Three tree-paths},
  journal = {Journal of Graph Theory},
  volume = {13},
  number = {2},
  pages = {175-188},
  doi = {https://doi.org/10.1002/jgt.3190130205},
  url = {https://onlinelibrary.wiley.com/doi/abs/10.1002/jgt.3190130205},
  eprint = {https://onlinelibrary.wiley.com/doi/pdf/10.1002/jgt.3190130205},
  abstract = {Abstract Itai and Rodeh [3] have proved that for any 2-connected graph G and any vertex s ∈ G there are two spanning trees such that the paths from any other vertex to s on the trees are disjoint. In this paper the result is generalized to 3-connected graphs.},
  year = {1989},
  timestamp = {Fri, 02 Oct 2020 01:00:00 +0200},
  biburl = {https://dblp.org/rec/journals/jgt/ZehaviI89.bib},
  bibsource = {dblp computer science bibliography, https://dblp.org},
  _bib2doi_selected = {dblp:/rec/journals/jgt/ZehaviI89.bib},
  _bib2doi_confirmed = {true},
  _bib2doi_finished = {true},
}

@article{general_edge,
  title = {On the generalized (edge-)connectivity of graphs},
  author = {Xueliang Li and Yaping Mao and Yuefang Sun},
  journal = {Australas. {J} Comb.},
  year = {2014},
  volume = {58},
  pages = {304--319},
  url = {http://ajc.maths.uq.edu.au/pdf/58/ajc\_v58\_p304.pdf},
  timestamp = {Wed, 11 Mar 2020 00:00:00 +0100},
  biburl = {https://dblp.org/rec/journals/ajc/LiMS14.bib},
  bibsource = {dblp computer science bibliography, https://dblp.org},
  _bib2doi_selected = {dblp:/rec/journals/ajc/LiMS14.bib},
  _bib2doi_confirmed = {true},
  _bib2doi_finished = {true},
}

@book{Steiner_book,
  title = {The {S}teiner Tree Problem},
  isbn = {9783322802910},
  issn = {0932-7134},
  url = {http://dx.doi.org/10.1007/978-3-322-80291-0},
  doi = {10.1007/978-3-322-80291-0},
  journal = {Advanced Lectures in Mathematics},
  publisher = {Vieweg+Teubner Verlag},
  author = {Pr\"{o}mel, Hans J\"{u}rgen and Steger, Angelika},
  year = {2002},
  _bib2doi_finished = {true},
}

@article{rainbow_trees,
  author = {Gary Chartrand and Futaba Okamoto and Ping Zhang},
  title = {Rainbow trees in graphs and generalized connectivity},
  journal = {Networks},
  volume = {55},
  number = {4},
  pages = {360--367},
  keywords = {rainbow tree, rainbow coloring, rainbow index, internally disjoint trees connecting a set of vertices, k-connectivity},
  doi = {10.1002/net.20339},
  url = {https://doi.org/10.1002/net.20339},
  eprint = {https://onlinelibrary.wiley.com/doi/pdf/10.1002/net.20339},
  abstract = {Abstract An edge-colored tree T is a rainbow tree if no two edges of T are assigned the same color. Let G be a nontrivial connected graph of order n and let k be an integer with 2 ≤ k ≤ n. A k-rainbow coloring of G is an edge coloring of G having the property that for every set S of k vertices of G, there exists a rainbow tree T in G such that S ⊆ V(T). The minimum number of colors needed in a k-rainbow coloring of G is the k-rainbow index of G. For every two integers k and n ≥ 3 with 3 ≤ k ≤ n, the k-rainbow index of a unicyclic graph of order n is determined. For a set S of vertices in a connected graph G of order n, a collection {T1,T2,…,Tℓ} of trees in G is said to be internally disjoint connecting S if these trees are pairwise edge-disjoint and V(Ti) ∩ V(Tj) = S for every pair i,j of distinct integers with 1 ≤ i,j ≤ ℓ. For an integer k with 2 ≤ k ≤ n, the k-connectivity κk(G) of G is the greatest positive integer ℓ for which G contains at least ℓ internally disjoint trees connecting S for every set S of k vertices of G. It is shown that κk(Kn)=n−⌈k/2⌉ for every pair k,n of integers with 2 ≤ k ≤ n. For a nontrivial connected graph G of order n and for integers k and ℓ with 2 ≤ k ≤ n and 1 ≤ ℓ ≤ κk(G), the (k,ℓ)-rainbow index rxk,ℓ(G) of G is the minimum number of colors needed in an edge coloring of G such that G contains at least ℓ internally disjoint rainbow trees connecting S for every set S of k vertices of G. The numbers rxk,ℓ(Kn) are determined for all possible values k and ℓ when n ≤ 6. It is also shown that for ℓ ϵ {1, 2}, rx3,ℓ(Kn) = 3 for all n ≥ 6. © 2009 Wiley Periodicals, Inc. NETWORKS, 2010},
  year = {2010},
  timestamp = {Sun, 28 May 2017 01:00:00 +0200},
  biburl = {https://dblp.org/rec/journals/networks/ChartrandOZ10.bib},
  bibsource = {dblp computer science bibliography, https://dblp.org},
  _bib2doi_old_doi = {https://doi.org/10.1002/net.20339},
  _bib2doi_selected = {dblp:/rec/journals/networks/ChartrandOZ10.bib},
  _bib2doi_confirmed = {true},
  _bib2doi_finished = {true},
}

@article{Cheriyan_directed,
  title = {Hardness and Approximation Results for Packing {S}teiner Trees},
  volume = {45},
  issn = {1432-0541},
  url = {https://doi.org/10.1007/s00453-005-1188-4},
  doi = {10.1007/s00453-005-1188-4},
  number = {1},
  journal = {Algorithmica},
  publisher = {Springer Science and Business Media LLC},
  author = {Joseph Cheriyan and Mohammad R. Salavatipour},
  year = {2006},
  month = {may},
  pages = {21--43},
  timestamp = {Wed, 17 May 2017 01:00:00 +0200},
  biburl = {https://dblp.org/rec/journals/algorithmica/CheriyanS06.bib},
  bibsource = {dblp computer science bibliography, https://dblp.org},
  _bib2doi_selected = {dblp:/rec/journals/algorithmica/CheriyanS06.bib},
  _bib2doi_confirmed = {true},
  _bib2doi_finished = {true},
}

@misc{directed_pendant,
  title = {Internally-disjoint Pendant {S}teiner Trees in Digraphs},
  author = {Shanshan Yu and Yuefang Sun},
  year = {2026},
  eprint = {2505.00298},
  archiveprefix = {arXiv},
  primaryclass = {math.CO},
  url = {https://arxiv.org/abs/2505.00298},
  doi = {10.48550/arXiv.2505.00298},
  _bib2doi_finished = {true},
}

@article{Huck1994,
  title = {Independent trees in graphs},
  volume = {10},
  issn = {1435-5914},
  url = {https://doi.org/10.1007/BF01202468},
  doi = {10.1007/BF01202468},
  number = {1},
  journal = {Graphs Comb.},
  publisher = {Springer Science and Business Media LLC},
  author = {Andreas Huck},
  year = {1994},
  month = {mar},
  pages = {29--45},
  timestamp = {Thu, 04 Jun 2020 01:00:00 +0200},
  biburl = {https://dblp.org/rec/journals/gc/Huck94.bib},
  bibsource = {dblp computer science bibliography, https://dblp.org},
  _bib2doi_old_doi = {10.1007/bf01202468},
  _bib2doi_selected = {dblp:/rec/journals/gc/Huck94.bib},
  _bib2doi_confirmed = {true},
  _bib2doi_finished = {true},
}

@article{kriesell,
  title = {Edge-disjoint trees containing some given vertices in a graph},
  journal = {J. Comb. Theory {B}},
  volume = {88},
  number = {1},
  pages = {53--65},
  year = {2003},
  issn = {0095-8956},
  doi = {10.1016/S0095-8956(02)00013-8},
  url = {https://doi.org/10.1016/S0095-8956(02)00013-8},
  author = {Matthias Kriesell},
  keywords = {Edge-connectivity, Spanning tree, Steiner tree, Connected factor, Line graph, Domination, Domatic number},
  abstract = {We show that for any two natural numbers k,ℓ there exist (smallest natural numbers fℓ(k)(gℓ(k)) such that for any fℓ(k)-edge-connected (gℓ(k)-edge-connected) vertex set A of a graph G with |A|⩽ℓ(|V(G)−A|⩽ℓ) there exists a system T of k edge-disjoint trees such that A⊆V(T) for each T∈T. We determine f3(k)=⌊8k+36⌋. Furthermore, we determine for all natural numbers ℓ,k the smallest number fℓ∗(k) such that every fℓ∗(k)-edge-connected graph on at most ℓ vertices contains a system of k edge-disjoint spanning trees, and give applications to line graphs.},
  timestamp = {Fri, 07 Jun 2024 01:00:00 +0200},
  biburl = {https://dblp.org/rec/journals/jct/Kriesell03.bib},
  bibsource = {dblp computer science bibliography, https://dblp.org},
  _bib2doi_old_doi = {https://doi.org/10.1016/S0095-8956(02)00013-8},
  _bib2doi_selected = {dblp:/rec/journals/jct/Kriesell03.bib},
  _bib2doi_confirmed = {true},
  _bib2doi_finished = {true},
}

@article{chartrand_general,
  author = {G. Chartrand and S. F. Kapoor and L. Lesniak and D. R. Lick},
  title = {Generalized connectivity in graphs},
  journal = {Bull. Bombay Math. Colloq.},
  volume = {2},
  year = {1984},
  pages = {1--6},
  mrnumber = {2107429},
  _bib2doi_finished = {true},
}

@inproceedings{edge_best,
  title = {Five edge-independent spanning trees},
  journal = {Procedia Computer Science},
  volume = {223},
  pages = {223--230},
  year = {2023},
  note = {XII Latin-American Algorithms, Graphs and Optimization Symposium (LAGOS 2023)},
  issn = {1877-0509},
  doi = {10.1016/j.procs.2023.08.232},
  url = {https://doi.org/10.1016/j.procs.2023.08.232},
  author = {Alonso Ali and Orlando Lee},
  keywords = {edge-independent spanning trees, edge-connectivity, graph decomposition},
  abstract = {Let G be a graph and let r be a fixed vertex of G. Two spanning trees T1 and T2 of G rooted at r are edge-independent if for every vertex v ϵ V(G), the paths from v to r in T1 and from v to r in T2 are edge-disjoint. Itai and Zehavi conjectured that for every k-edge-connected graph and any vertex r ϵ V(G) there are k edge-independent spanning trees rooted at r (Edge-Independent Spanning Trees Conjecture). Itai and Rodeh proved the case k = 2, Schlipf and Schmidt proved the case k = 3, and Hoyer and Thomas proved the case k = 4 of the conjecture. In this paper, we prove the case k = 5.},
  timestamp = {Sat, 08 Jun 2024 01:00:00 +0200},
  biburl = {https://dblp.org/rec/conf/lagos/AliL23.bib},
  bibsource = {dblp computer science bibliography, https://dblp.org},
  booktitle = {Proceedings of the {XII} Latin-American Algorithms, Graphs and Optimization Symposium, {LAGOS} 2023, Huatulco, Mexico, September 18-22, 2023},
  publisher = {Elsevier},
  editor = {Cristina G. Fernandes and Sergio Rajsbaum},
  series = {Procedia Computer Science},
  _bib2doi_old_doi = {https://doi.org/10.1016/j.procs.2023.08.232},
  _bib2doi_selected = {dblp:/rec/conf/lagos/AliL23.bib},
  _bib2doi_confirmed = {true},
  _bib2doi_finished = {true},
}

@article{node_best,
  author = {Curran, Sean and Lee, Orlando and Yu, Xingxing},
  title = {Finding Four Independent Trees},
  journal = {SIAM Journal on Computing},
  volume = {35},
  number = {5},
  pages = {1023-1058},
  year = {2006},
  doi = {10.1137/S0097539703436734},
  url = {https://doi.org/10.1137/S0097539703436734},
  eprint = {https://doi.org/10.1137/S0097539703436734},
  abstract = {Motivated by a multitree approach to the design of reliable communication protocols, Itai and Rodeh gave a linear time algorithm for finding two independent spanning trees in a 2-connected graph. Cheriyan and Maheshwari gave an \$O(|V|^2)\$ algorithm for finding three independent spanning trees in a 3-connected graph. In this paper we present an \$O(|V|^3)\$ algorithm for finding four independent spanning trees in a 4-connected graph. We make use of chain decompositions of 4-connected graphs.},
  timestamp = {Sat, 27 May 2017 01:00:00 +0200},
  biburl = {https://dblp.org/rec/journals/siamcomp/CurranLY06.bib},
  bibsource = {dblp computer science bibliography, https://dblp.org},
  _bib2doi_selected = {dblp:/rec/journals/siamcomp/CurranLY06.bib},
  _bib2doi_confirmed = {true},
}

@article{locally_twisted,
  title = {Improving the diameters of completely independent spanning trees in locally twisted cubes},
  journal = {Inf. Process. Lett.},
  volume = {141},
  pages = {22--24},
  year = {2019},
  issn = {0020-0190},
  doi = {10.1016/j.ipl.2018.09.006},
  url = {https://doi.org/10.1016/j.ipl.2018.09.006},
  author = {Kung{-}Jui Pai and Jou{-}Ming Chang},
  keywords = {Interconnection networks, Completely independent spanning trees, Locally twisted cubes, Diameter},
  abstract = {Let G be a graph and denote diam(G) the diameter of G, i.e., the greatest distance between any pair of vertices in G. A set of spanning trees of G are called completely independent spanning trees (CISTs for short) if for every pair of vertices x,y∈V(G), the paths joining x and y in any two trees have neither vertex nor edge in common, except x and y. Pai and Chang (2016) [12] recently proposed a unified approach to recursively construct two CISTs in several hypercube-variant networks, including locally twisted cubes. For every kind of n-dimensional variant cube, the diameters of two CISTs for their construction are 2n−1. In this note, we provide a new scheme to construct two CISTs T1 and T2 in locally twisted cubes LTQn, and thereby prove the following improvement: for i∈{1,2}, diam(Ti)=2n−2 if n=4; and diam(Ti)=2n−3 if n⩾5. In particular, the construction of CISTs for LTQ4 is optimal in the sense of diameter.},
  timestamp = {Fri, 02 Nov 2018 00:00:00 +0100},
  biburl = {https://dblp.org/rec/journals/ipl/PaiC19.bib},
  bibsource = {dblp computer science bibliography, https://dblp.org},
  _bib2doi_old_doi = {https://doi.org/10.1016/j.ipl.2018.09.006},
  _bib2doi_selected = {dblp:/rec/journals/ipl/PaiC19.bib},
  _bib2doi_confirmed = {true},
  _bib2doi_finished = {true},
}

@article{edmonds1973edge,
  title = {Edge-disjoint branchings},
  author = {Edmonds, Jack},
  journal = {Combinatorial algorithms},
  pages = {91--96},
  year = {1973},
  publisher = {Academic Press},
  _bib2doi_finished = {true},
}

@techreport{berczi2009packing,
  author = {B\'erczi, Krist\'of and Frank, Andr\'as},
  title = {Packing Arborescences},
  institution = {Egerv\'ary Research Group, E\"otv\"os University},
  year = {2009},
  number = {TR-2009-04},
  month = {May},
  url = {https://egres.elte.hu/tr/egres-09-04.pdf},
  _bib2doi_finished = {true},
}

@misc{szigeti2017survey,
  author = {Szigeti, Zolt\'{a}n},
  title = {Packing Arborescences: A Survey},
  year = {2017},
  month = {April 20},
  howpublished = {Slides, G-SCOP, Univ. Grenoble Alpes},
  url = {https://pagesperso.g-scop.grenoble-inp.fr/~szigetiz/slides/tokyo_1.pdf},
  _bib2doi_finished = {true},
}


\end{document}